\documentclass{lmcs}
\pdfoutput=1
\usepackage[utf8]{inputenc}

\usepackage{lastpage}
\lmcsdoi{22}{2}{15}
\lmcsheading{}{\pageref{LastPage}}{}{}%
{Jan.~31,~2024}{May~06,~2026}{}

\usepackage{amssymb}
\usepackage[all]{xy}

\newcommand{\da}{{\downarrow}}
\newcommand{\Var}{\mathit{Var}}

\newcommand{\fn}{\mathit{fn}}

\newcommand{\unblind}{\mathit{unblind}}
\newcommand{\blind}{\mathit{blind}}
\newcommand{\sign}{\mathit{sign}}
\newcommand{\checksign}{\mathit{checksign}}
\newcommand{\dec}{\mathit{dec}}
\newcommand{\enc}{\mathit{enc}}
\newcommand{\mal}{\mathit{mal}}
\newcommand{\plus}{\mathit{plus}}
\newcommand{\pred}{\mathit{pred}}
\newcommand{\prefix}{\mathit{prefix}}
\newcommand{\fst}{\mathit{fst}}
\newcommand{\snd}{\mathit{snd}}
\newcommand{\thd}{\mathit{thd}}
\newcommand{\open}{\mathit{open}}
\newcommand{\td}{\mathit{td}}
\newcommand{\adec}{\mathit{adec}}
\newcommand{\aenc}{\mathit{aenc}}
\newcommand{\pk}{\mathit{pk}}
\newcommand{\chck}{\mathit{check}}
\newcommand{\msg}{\mathit{msg}}
\newcommand{\ok}{\mathit{ok}}

\keywords{Term Rewriting, Security Protocols, Deduction, Static Equivalence}

\begin{document}
	

\title[Extending the Notion of Subterm Convergent]{Knowledge Problems in Protocol Analysis:\\ Extending the Notion of Subterm Convergent}

	\thanks{We would like to thank the anonymous reviewers of this article for their helpful comments, feedback, and suggestions. Their input has improved the paper. We also thank Paliath Narendran for his helpful input.}	
	
	\author[C.~Bunch]{Carter Bunch}[a]
	\author[S.~Dwyer Satterfield]{Saraid Dwyer Satterfield}[a]
	\author[S.~Erbatur]{Serdar Erbatur\lmcsorcid{0000-0002-7574-195X}}[b]
	\author[A.M.~Marshall]{Andrew M. Marshall\lmcsorcid{0000-0002-0522-8384}}[a]
	\author[C.~Ringeissen]{Christophe Ringeissen\lmcsorcid{0000-0002-5937-6059}}[c]
	
	\address{University of Mary Washington, Fredericksburg, VA, USA}	
	
	\address{University of Texas at Dallas, Richardson, TX, USA}	
	
	\address{Universit\'e de Lorraine, CNRS, Inria, LORIA, F-54000 Nancy, France}	

	\begin{abstract}
		We introduce a new form of restricted term rewrite system, the graph-embedded term rewrite system. These systems, and thus the name, are inspired by the graph minor relation and are more flexible extensions of the well-known homeomorphic-embedded property of term rewrite systems. As a motivating application area, we consider the symbolic analysis of security protocols, and more precisely the two knowledge problems defined by the deduction problem and the static equivalence problem. In this field, restricted term rewrite systems, such as subterm convergent ones, have proven useful since the knowledge problems are decidable for such systems. Many of the same decision procedures still work for examples of systems which are ``beyond subterm convergent''. However, the applicability of the corresponding decision procedures to these examples must often be  proven on an individual basis. This is due to the problem that they do not fit into an existing syntactic definition for which the procedures are known to work. Here we show that many of these systems belong to a particular subclass of graph-embedded convergent systems, called contracting convergent systems. On the one hand, we show that the knowledge problems are decidable for the subclass of contracting convergent systems. On the other hand, we show that the knowledge problems are undecidable for the class of graph-embedded systems. Going further, we compare and contrast these graph embedded systems with several notions and properties already known in the protocol analysis literature. 
		Finally, we provide several combination results, both for the combination of multiple contracting convergent systems, and then for the combination of contracting convergent systems with particular permutative equational theories. 
	\end{abstract}
	
	\maketitle        
	\section{Introduction}
	In this paper we introduce a new form of term rewrite system, called the graph-embedded term rewrite systems, and motivate the study and use of such rewrite systems by demonstrating their usefulness in the application of security protocols. 
	
	The research area of cryptographic protocol analysis contains a number of innovative algorithms and procedures for checking various security properties of protocols, 
	see for example~\cite{DBLP:journals/tcs/AbadiC06, DBLP:journals/tocl/BaudetCD13, DBLP:journals/tocl/ChadhaCCK16, DBLP:journals/jar/CiobacaDK12, Dreier2017}. These procedures consider protocols modeled in a symbolic way, typically via a rewrite system or equational theory. Often the procedure is proven sound and complete for specific classes  of theories. One of the most common classes are those theories that can be represented by subterm convergent term rewrite systems. That is, term rewrite systems where the right-hand side of the rules are strict subterms of the left-hand side or a constant. For example, see the procedures developed in~\cite{DBLP:journals/tcs/AbadiC06, DBLP:journals/jar/CiobacaDK12}. Interestingly, many of these same procedures also work for theories that are ``beyond subterm convergent''. That is, they are not strictly subterm convergent.  However, since these examples do not fit into a known class of theories for which soundness and completeness proofs already exist, they must be proven on an individual bases.  For example, the procedures of~\cite{DBLP:journals/tcs/AbadiC06, DBLP:journals/tocl/BaudetCD13, DBLP:journals/tocl/ChadhaCCK16, DBLP:journals/jar/CiobacaDK12, Dreier2017} are shown to work on the theory of blind signatures, see Example~\ref{example:blind-sig} below. However, the theory is not subterm convergent, notice in the final rule,
	$\unblind(\sign(\blind(x, y), z), y) \rightarrow \sign(x, z)$, that $\sign(x, z)$ is not a strict subterm of 
	$\unblind(\sign(\blind(x, y), z), y)$. Thus, in each case a unique proof is needed to show applicability of the procedure on the theory of blind signatures. Several additional examples of beyond subterm theories are given throughout this paper. This begs the question of whether there is a syntactic definition of a class of term rewrite systems such that the definition encapsulates these beyond subterm examples yet still maintains some of the useful properties needed to ensure applicability of the above procedures. 
	
	In this paper we answer the question in the positive by introducing first graph-embedded term rewrite systems and then a particular subclass called contracting rewrite systems. These systems are inspired by the notions of graph embeddings and graph minors. Here we are able to translate the notion to term rewrite systems. This translation is done in a very similar fashion to what has been done with  homeomorphic embeddings. We are able to provide a rewrite schema which induces graph-embedded systems in a similar way in which homeomorphic-embedded systems are induced by a rewrite system (see~\cite{Baader98} for more details).
	To the best of our knowledge these systems have not been explored before. We then explore some of the properties of these new systems. Interestingly, the graph-embedded systems encompass most of the beyond subterm examples from many of the protocol analysis procedure papers~\cite{DBLP:journals/tcs/AbadiC06, DBLP:journals/tocl/BaudetCD13,  DBLP:journals/tocl/ChadhaCCK16, DBLP:journals/jar/CiobacaDK12, Dreier2017}. As an initial step, in this paper we concentrate on the knowledge problems considered in~\cite{DBLP:journals/tcs/AbadiC06} using the notion of locally stable theories. Local stability is a desirable property which ensures the decidability of the critical symbolic security question of deducibility.
	In the class of graph-embedded convergent systems, we are now able to identify a particular subclass called the contracting convergent systems, which are beyond subterm convergent, encompass most of the beyond subterm examples of~\cite{DBLP:journals/tcs/AbadiC06, DBLP:journals/tocl/BaudetCD13, DBLP:journals/tocl/ChadhaCCK16, DBLP:journals/jar/CiobacaDK12, Dreier2017}, and are locally stable. As a consequence, the knowledge problems of deduction and static equivalence are decidable for the subclass of contracting convergent systems. We show that the knowledge problems are undecidable for the class of graph-embedded convergent systems in general. Going further, we also compare the graph-embedded systems
	to the Finite Variant Property (FVP), another useful property in protocol analysis, and show that the two do not define the same set of term rewrite systems but that a sub-class of graph embedding system, a restricted form of the contracting systems, can be defined 
	which always guarantees the FVP. We consider the cap problem~\cite{DBLP:conf/rta/AnantharamanNR07}, a related knowledge problem, showing it is also decidable for contracting systems, when certain signature requirements are met. In addition to the procedure in~\cite{DBLP:journals/tcs/AbadiC06} we also investigate graph-embedded term rewrite systems with
	the YAPA procedure~\cite{DBLP:journals/tocl/BaudetCD13}, showing that contracting graph-embedded term rewrite systems have the layered property, an important property for guaranteeing success of the YAPA procedure. Lastly we present several combination results for contracting convergent term rewrite systems. 
	
	This paper represents an exploration of graph-embedded term rewrite systems and their application to protocol analysis. We hope that the formulation proves useful in areas beyond security protocols as homeomorphic embeddings have proven useful in many areas. We conclude the paper with a discussion of several open questions related to graph-embedded systems. 
	
	An initial version of these results was presented in~\cite{DBLP:conf/fscd/SatterfieldEMR23}. The current paper expands on~\cite{DBLP:conf/fscd/SatterfieldEMR23} by including a number of new results
	and answering several questions left open in~\cite{DBLP:conf/fscd/SatterfieldEMR23}, specifically:
	\begin{itemize}
		\item We consider the cap problem, left open in~\cite{DBLP:conf/fscd/SatterfieldEMR23}, showing new decidability results (see Section~\ref{sec:relations-cap}).
		\item We expand on the FVP results of~\cite{DBLP:conf/fscd/SatterfieldEMR23} by including new positive results (see Section~\ref{sec:properties}).
		\item We answer a question on graph-embedded systems and the YAPA 
		procedure left open in~\cite{DBLP:conf/fscd/SatterfieldEMR23}, showing that a restricted form of graph embedded systems are layered (see Section~\ref{sec:relations-layered}).
		\item We have included new combination results between contracting convergent systems, which were not considered in~\cite{DBLP:conf/fscd/SatterfieldEMR23} (see Section~\ref{sec:combination}).
		\item Finally, we have made several improvements on the material of~\cite{DBLP:conf/fscd/SatterfieldEMR23}. We have simplified the definition of contracting convergent TRS, making the definition easier to understand. We have also fixed an issue with the proof of the undecidability of the knowledge problems in graph-embedded systems presented in~\cite{DBLP:conf/fscd/SatterfieldEMR23}, correcting a problem that allowed trivial solutions. 
	\end{itemize}

	\paragraph*{Paper Outline.}
	The remainder of the paper is organized as follows. Section~\ref{sec:prelim} contains the preliminaries, introducing the necessary background material on term-rewrite systems, graph theory and security protocol analysis.
	Section~\ref{sec:graphemb} introduces the graph-embedded term rewrite systems and explores some of their basic properties. Then, the next sections are related to the motivating application area of this paper for graph-embedded systems, security protocol analysis. The knowledge problems are shown to be undecidable for the class of graph-embedded convergent systems (Section~\ref{sec:undecidable-know}) and decidable for the subclass of contracting convergent systems (Section~\ref{sec:decidable-know}).
	Section~\ref{sec:relations} considers the relation to the cap problem, to the FVP, and to the YAPA procedure. Section~\ref{sec:combination} considers the combination question in contracting convergent systems. Finally, Section~\ref{sec:conclusions} contains the concluding remarks, future work, and open problems. 
		
	\section{Preliminaries}\label{sec:prelim}
	We use the standard notation of equational
	unification~\cite{BaaderSnyd-01} and term rewriting
	systems~\cite{Baader98}.  Given a first-order signature $\Sigma$ and a
	(countable) set of variables $V$, the $\Sigma$-terms over
	variables $V$ are built in the usual way by taking into account the arity of each function symbol in $\Sigma$. Arity $0$ function symbols are called \emph{constants}. Each $\Sigma$-term is well-formed: if it is rooted by a $n$-ary function symbol in $\Sigma$, then it has necessarily $n$ direct subterms. 
	The set of $\Sigma$-terms over variables $V$ is denoted by $T(\Sigma,V)$. 
	The set of variables occurring in a term $t$ is denoted by $\Var(t)$.
        A term $t$ is \emph{ground} if $\Var(t) = \emptyset$.
        A $\Sigma$-rooted term is a term whose root symbol
	is in $\Sigma$. For any position~$p$ in a term $t$ (including
	the root position $\epsilon$), $t(p)$ is the symbol at position $p$,
	$t|_p$ is the subterm of $t$ at position $p$, and $t[u]_p$ is the term
	$t$ in which $t|_p$ is replaced by $u$.  A substitution is an
	endomorphism of $T(\Sigma,V)$ with only finitely many variables not
	mapped to themselves. A substitution is denoted by $\sigma = \{ x_1
	\mapsto t_1, \dots, x_m \mapsto t_m \}$, where the domain of $\sigma$
	is $Dom(\sigma) = \{ x_1,\dots,x_m \}$ and the range of $\sigma$ is
	$Ran(\sigma) = \{ t_1,\dots,t_m \}$.  Application of a substitution
	$\sigma$ to $t$ is written $t\sigma$. 

        Given a term $t$ and a function symbol or variable $\mathit{fv}$, $|t|_{\mathit{fv}}$ denotes the number of occurrences of $\mathit{fv}$ in $t$. 
	The size of a term $t$, denoted by $|t|$, is defined inductively as follows: $|f(t_1,\dots,t_n)| = 1 + \Sigma_{i=1}^n |t_i|$ if $f$
	is a $n$-ary function symbol with $n\geq 1$, $|c| = 1$ if $c$ is a
	constant, and $|x| = 1$ if $x$ is a variable. 
        The depth of a term $t$, denoted by $depth(t)$, is defined inductively as follows: $depth(f(t_1,\dots,t_n)) = 1 + \max_{i=1,\dots, n} depth(t_i)$ if $f$ a $n$-ary function symbol with $n\geq 1$, $depth(c) = 0$ if $c$ is a constant, and $depth(x) = 0$ if $x$ is a variable.

        For any term $t$, let $st(t)$ be the set of subterms of $t$, and given any set of terms $T$, let $st(T) = \bigcup_{t \in T} st(t)$.

	Let $V_{C}$ denote a finite set of context variables (sometimes called holes) such that $V_{C} \cap V = \emptyset$. We use the notation $\diamond_i$ to represent the context variable $\diamond_i \in V_{C}$. A \emph{context} is a \emph{linear} term containing only context variables. More formally, a context is a term, $C \in T(\Sigma, V_{C})$, where each variable (context hole) occurs at most once. Thus, the size of a context follows from the size of a term, where any hole occurrence counts for $1$.
        Given a context $C$ with $m$ variables and $m$ terms $S_1,\dots,S_m \in T(\Sigma, V)$, $C[S_1,\dots,S_m]$ denotes the term $C\{ \diamond_1 \mapsto S_1,\dots,\diamond_m \mapsto S_m \}$ where for $i=1,\dots,m$, $\diamond_i$ denotes the $i$-th context variable occurring in the term $C$, and $C$ is called the \emph{context part} of $C[S_1,\dots,S_m]$. When $S_1,\dots,S_m \in \mathcal{S}$, $C[S_1,\dots,S_m]$ is said to be a \emph{context instantiated with terms in $\mathcal{S}$}. Given a positive integer $n$, $C[S_1,\dots,S_m]$ is said to be \emph{small} (with respect to $n$) if its context part $C$ satisfies $|C| \leq n$.

	\paragraph*{Equational Theories}
	
	Given a set $E$ of
	$\Sigma$-axioms (i.e., pairs of terms in $T(\Sigma,V)$, denoted by $l = r$),
	the \emph{equational theory} $=_E$ is the congruence closure of $E$
	under the law of substitutivity (by a slight abuse of terminology, $E$
	is often called an equational theory).
	Equivalently, $=_E$ can be defined as the reflexive transitive closure $\leftrightarrow_{E}^*$ of an equational step $\leftrightarrow_{E}$ defined as follows:
	$s \leftrightarrow_{E} t$ if there exist a
	position $p$ of $s$, $l = r$  (or $r=l$) in $E$, and substitution
	$\sigma$ such that $s|_p = l\sigma$ and $t = s[r\sigma]_p$. An equational theory $E$ is said to be \emph{permutative} if for any $l=r$ in $E$, the number of occurrences of any (function or variable) symbol in $l$ is equal to the number of occurrences of that symbol in $r$. Well-known theories such as Associativity ($A = \{
(x+y)+z = x+(y+z) \}$), Commutativity ($C = \{ x + y = y + x \}$), and
Associativity-Commutativity ($AC = A \cup C$) are permutative
theories. A theory $E$ is said to be \emph{shallow} if variables can only occur at a depth at most $1$ in axioms of $E$. For example, $C$ is shallow but $A$ and $AC$ are not. Any constant in $\Sigma$ not occurring in $E$ is said to be \emph{free}.

	\paragraph*{Rewrite Relations}
	
	A \emph{term rewrite system} (TRS) is a pair $(\Sigma, R)$, where
	$\Sigma$ is a signature and $R$ is a finite set of rewrite rules of
	the form $l \rightarrow r$ such that $l,r$ are $\Sigma$-terms, $l$ is
	not a variable and $\Var(r) \subseteq \Var(l)$.
        A function symbol $f$ is a \emph{constructor} for $R$ if $\{f\} \cap \{ l(\epsilon) ~|~ l \rightarrow r \in R \} = \emptyset$.
	A term $s$ \emph{rewrites} to a term $t$ w.r.t $R$, denoted
	by $s \rightarrow_R t$ (or simply $s \rightarrow t$), if there exist a
	position $p$ of $s$, $l \rightarrow r \in R$, and substitution
	$\sigma$ such that $s|_p = l\sigma$ and $t = s[r\sigma]_p$.
	If the rewrite step occurs at the root position (resp., at some non-rooted position) of the term $s$ we denote this as $s \rightarrow_{R}^{\epsilon} t$ (resp., $s\stackrel{\neq \epsilon}{\rightarrow_R} t$).  
        When $\sigma$ is a variable renaming, we say that $s$ rewrites to $t$ applying a \emph{variable instance} of $l \rightarrow r$.
        The reflexive transitive closure of $\rightarrow_R$ is denoted by $\rightarrow_R^*$. We write $\rightarrow_R^{0,1}$ to denote $0$ or $1$ rewrite step w.r.t $R$. The composition of relations is denoted by the $(.)$ operator; for instance, $\rightarrow_R . \rightarrow_R$ is defined as follows: for any $s,u$, $s \rightarrow_R . \rightarrow_R u$ if there exists some $t$ such that $s \rightarrow_R t \rightarrow_R u$.
        A TRS $R$ is \emph{terminating} if there are no infinite rewriting sequences with respect to $\rightarrow_R$.  A TRS $R$
	is \emph{confluent} if, whenever $t \rightarrow_R^{*} s_1$ and $t \rightarrow_R^{*} s_2$, there exists a term $w$ such that $s_1 \rightarrow_R^{*} w$ and
	$s_2 \rightarrow_R^{*} w$.  A confluent and terminating TRS is
	called \emph{convergent}. In a convergent TRS $R$, we have the existence and
	the uniqueness of $R$-normal forms, denoted by $t\da_R$ for any
	term $t$. When $R$ is clear from the context, the normal form of $t$ may be written $t\da$. Given a substitution $\sigma$, $\sigma\da = \{ x \mapsto (x\sigma)\da \}_{x \in Dom(\sigma)}$ is the substitution corresponding to the normal form of $\sigma$. A substitution $\sigma$ is in \emph{$R$-normal form} (or \emph{$R$-normalized}) if $\sigma = \sigma\da$. A rewrite derivation w.r.t $R$ ended by a $R$-normal form is denoted by $\rightarrow_R^!$.

	A rewrite rule is said to be {\em subterm} if its right-hand side is either a strict subterm of its left-hand side or a constant. A TRS $R$ is {\em subterm} if any rule in $R$ is subterm. A TRS $R$ is said to be {\em subterm convergent} if it is both convergent and subterm. An equational theory, $E$, is {\em
		subterm convergent} if it is presented by a subterm convergent
	TRS. That is, there exists a subterm convergent TRS, $R$, such that $=_E$
	and $=_R$ coincide. By a slight abuse of notation, the equational theory of $R$, given by $\{ l = r ~|~ l \rightarrow r \in R \}$, may also be denoted by $R$. The \emph{size} of a TRS $R$ is defined as 
	\[c_{R} = max_{\{ l ~|~ l \rightarrow r \in R\}}(|l|, ar(R) + 1)\] 
	where $ar(R)$ is the maximal arity of any function symbol occurring in $R$.  
	
	\begin{rem}
		We should note that in some papers the definition of subterm convergent is expanded to also include ground terms on the right-hand side. We will use a more classical definition here, allowing only a right-hand side which is either a strict subterm of the left-hand side or a constant.
	\end{rem}       
	
	\begin{defi}[Homeomorphic Embedding]\label{def:hom-emd}
		The \emph{homeomorphic embedding}, $\trianglerighteq_{emb}$ is a 
		binary relation on terms such that: 
		$s \trianglerighteq_{emb} t$ if one of the following conditions hold:
		\begin{enumerate}
			\item $s = x = t$ for some variable $x$,
			\item $s = f(s_1, \ldots, s_n)$ and 
			$t=f(t_1, \ldots, t_n)$ and 
			$s_1 \trianglerighteq_{emb} t_1, \ldots, 
			s_n \trianglerighteq_{emb} t_n$,
			\item $s = f(s_1, \ldots, s_n)$ and 
			$s_i \trianglerighteq_{emb} t$ for some $i$,  
			$1 \leq i \leq n$.
		\end{enumerate}
		A TRS $R$ is said to be a {\em homeomorphic-embedded} TRS if for any $l \rightarrow r \in R$, $l \trianglerighteq_{emb} r$.
		
	\end{defi}
	
	More interestingly we can also define $\trianglerighteq_{emb}$
	as the reduction relation $\rightarrow^{*}_{R_{emb}}$ 
	induced by the rewrite system
	$R_{emb} = \{f(x_1, \ldots, x_n) \rightarrow x_i ~| 
	~ f \mbox{ is $n$-ary}, n \geq 1, ~1 \leq i \leq n \}$. 
	
	\begin{exa}[Blind Signatures]\label{example:blind-sig}
		The theory of blind signatures~\cite{DBLP:journals/jar/CiobacaDK12} is a\\ homeomorphic-embedded convergent TRS:
		\begin{align*}
			\checksign(\sign(x,y),\pk(y)) &\rightarrow x\\
			\unblind(\blind(x, y), y) &\rightarrow x \\
			\unblind(\sign(\blind(x, y), z), y) &\rightarrow \sign(x, z)
		\end{align*}
		
	\end{exa}

	\paragraph*{Notions of Knowledge}
	
	The applied pi calculus and frames are used to model attacker
	knowledge~\cite{Abadi2001}. In this model, the set of messages or
	terms which the attacker knows, and which could have been obtained
	from observing a protocol session, are the set of terms given by the frame $\phi = \nu \tilde{n}.\sigma$, where
	$\sigma$ is a substitution ranging over \emph{ground terms}. We also need to
	model cryptographic concepts such as nonces, keys, and publicly known
	values. We do this by considering terms built using a signature $\Sigma$ which is assumed to contain a countable set of free constants, including a countable set of names $N$. In this context, we need to track the names which the attacker knows, such as public values, and the names which the attacker does
	not know a priori, such as freshly generated nonces.
        In a frame $\nu \tilde{n}.\sigma$, $\tilde{n}$
	is a finite set of names  in $N$, these names represent
	freshly generated nonces which the attacker does not initially know.  The
	set of names occurring in a term $t$ is denoted by $\fn(t)$. For any frame $\phi=\nu \tilde{n}.\sigma$, let $\fn(\phi)$ be the set of names $\fn(\sigma) \backslash \tilde{n}$ where $\fn(\sigma) = \bigcup_{t \in Ran(\sigma)} \fn(t)$; and for any term $t$, let $t\phi$ denote by a slight abuse of notation the term $t\sigma$. We say that a term $t$ is {\em $\phi$-restricted}, or simply {\em restricted} when $\phi$ is clear from the context, if $\fn(t) \cap
	\tilde{n} = \emptyset$. Given any convergent TRS $R$, the frame $\phi$ is \emph{$R$-normalized} if $\sigma$ is $R$-normalized.

	\begin{defi}[Deduction]
		\label{def:deduction}  
		Let $\phi = \nu \tilde{n}.\sigma$ be a frame, and $t$ a ground term. 
		We say that $t$ is {\em deduced from $\phi$ modulo $E$}, denoted by 
		$\phi \vdash_E t$, if there exists a term $\zeta$ such that $\zeta\sigma =_{E} t$ and $\fn(\zeta) \cap \tilde{n} = \emptyset$. The term $\zeta$ is called a \emph{recipe} of $t$ in $\phi$ modulo $E$. 
	\end{defi}
	
	Another form of knowledge is the ability to tell if two frames are 
	\emph{statically equivalent} modulo $E$, sometimes also called 
	\emph{indistinguishability}.

	\begin{defi}[Static Equivalence]
		Two terms $s$ and $t$ are {\em equal} 
		in a frame $\phi =  \nu \tilde{n}.\sigma$ modulo an 
		equational theory $E$, denoted $(s =_E t) \phi$, if
		$s\sigma =_E t\sigma$, and 
		${ \tilde{n} } \cap (\fn(s)\cup \fn(t)) = \emptyset$.
		The set of all equalities $s=t$ such that $(s=_E t)\phi$ is denoted by $Eq(\phi)$. Given a set of equalities $Eq$, the fact that $(s=_E t)\phi$ for any $s = t \in Eq$ is denoted by $\phi \models Eq$.
		Two frames $\phi = \nu \tilde{n}.\sigma$ and $\psi= \nu \tilde{n}.
		\tau$ are {\em statically equivalent modulo $E$}, denoted as 
		$\phi \approx_E \psi$, if $Dom(\sigma) = Dom(\tau)$, $\phi \models Eq(\psi)$ and $\psi \models Eq(\phi)$.  
	\end{defi}
	
	\begin{rem}
		Again the above definition relies on the use of the frame definition from the pi calculus. One can view the set $\tilde{n}$ of bound names as those secrets the attack does not have initial knowledge of for any run of the protocol. Each of the frames thus represent runs of the protocols when the attacker can then compare. 
	\end{rem}

        Since the equational theory of $R$ may be simply denoted by $R$, we may write for instance $\phi \vdash_R t$ when $t$ is deduced from $\phi$ modulo (the equational theory of) $R$. 

        The local stability property was introduced in~\cite{DBLP:journals/tcs/AbadiC06} and improved in~\cite{DBLP:journals/tcs/Ayala-RinconFN17} as a collection of conditions which are sufficient to obtain the decidability of deduction thanks to a reduction to deduction modulo the empty theory. A simplified definition of this property is introduced below. It is simplified because we do not consider $AC$-symbols as in~\cite{DBLP:journals/tcs/AbadiC06, DBLP:journals/tcs/Ayala-RinconFN17}. 
        
	\begin{defi}[Local Stability~\cite{DBLP:journals/tcs/AbadiC06}]\label{def:locally-stable} A convergent TRS, $R$, is \emph{locally stable}
		if, for every $R$-normalized frame $\phi = \nu \tilde{n}.\sigma$, there exists a finite set $sat(\phi)$ of ground terms such that:
		\begin{itemize}
			\item $Ran(\sigma) \subseteq sat(\phi)$ and $\fn(\phi) \subseteq sat(\phi)$;
			\item if $M_1, \ldots, M_k \in sat(\phi)$ and $f(M_1, \ldots, M_k) \in st(sat(\phi))$, then
			$f(M_1, \ldots, M_k) \in sat(\phi)$;
			\item if $C[S_1, \ldots, S_l] \rightarrow_{R}^{\epsilon} M$, where $C$ is a context with $|C| \leq c_{R}$ and
			$\fn(C) \cap \tilde{n} = \emptyset$, and $S_1,\ldots, S_l \in sat(\phi)$, then there exist a context $C'$
			and $S_1',\ldots, S_k' \in sat(\phi)$ such that $|C'| \leq c_{R}^2$, 
			$\fn(C') \cap \tilde{n} = \emptyset$, and $M \rightarrow_{R}^* C'[S_1', \ldots, S_k']$;
			\item if $M \in sat(\phi)$ then $\phi \vdash_{R} M$.
		\end{itemize}
	\end{defi}

Beside local stability, the additional property of local finiteness is useful for the decidability of static equivalence.
In a locally finite theory, only finitely many equalities have to be checked for static equivalence. Helpfully, in a convergent TRS with no $AC$-symbols, local stability implies local finiteness~\cite{DBLP:journals/tcs/AbadiC06}.
Since we do not consider $AC$-symbols in this paper,  the approach presented in~\cite{DBLP:journals/tcs/AbadiC06} leads to the following result:

\begin{thmC}[\cite{DBLP:journals/tcs/AbadiC06}]
\label{thm:know-by-local-stability}
  The deduction and static equivalence problems are both decidable in any locally stable TRS.
\end{thmC}  

Since any subterm convergent TRS is locally stable, the decidability of both deduction and static equivalence in the class of subterm convergent TRSs follows from Theorem~\ref{thm:know-by-local-stability}. 
	\begin{exa}
		Let $E$ be the equational theory presented by the subterm convergent TRS $\{ \dec(\enc(x,y),y) \rightarrow x \}$. Notice that this theory is subterm convergent and contains no AC symbols. Thus, from Theorem~\ref{thm:know-by-local-stability} static equivalence is decidable. If we are given two frames, $\phi = \nu \{ n \}.\{ v \mapsto \enc(a,n) \}$ and $\psi = \nu \{ n\}.\{ v \mapsto \enc(b,n) \}$, then we can apply the saturation based procedure defined in~\cite{DBLP:journals/tcs/AbadiC06} to find that these two frames are statically equivalent modulo $E$. However, consider now $\phi' = \nu \{ n \}.\{ v \mapsto \enc(a,n), w \mapsto n \}$ and $\psi' = \nu \{ n \}.\{ v \mapsto \enc(b,n), w \mapsto n \}$. Since $\dec(v,w) =_{E} a \in Eq(\phi')$ but $\dec(v,w) =_{E} a \notin Eq(\psi')$, $\phi'$ and $\psi'$ are not statically equivalent modulo $E$.   
	\end{exa}

In this paper, we lift this result to rewrite systems that are beyond the class of subterm convergent rewrite systems.

	\paragraph*{Term Graphs}
	
	Each term $t$ can be viewed in a graphical representation, called a 
	\emph{term graph}. Each vertex in the graph is labeled either by a function symbol or a variable. Each function symbol vertex also has an associated successor number, corresponding to the arity of the function. Edges connect the vertices of the term graph based on the subterm relation. The notion of term graph is illustrated in Examples~\ref{example: graph-emb-ex1} and~\ref{example:mal-encrypt}.

	\paragraph*{Some Graph Theory}
	
	We will also need a few notions from graph theory, we introduce those in this section. We will typically use $G$
	to denote a graph, $V$ the set of vertices and $E \subseteq V \times V$ the set of edges of the graph. 
	
	\begin{defi}[Graph Isomorphism]
		Let $G= (V, E)$ and $G' = (V', E')$ be two graphs. We say that $G$ and $G'$ are isomorphic, 
		denoted $G \simeq G'$, if there exists a bijection $\phi :V \rightarrow V'$ with
		$xy \in E $ iff $\phi(x)\phi(y) \in E'$, $ \forall x, y \in V$.
	\end{defi}
	
	\begin{defi}[Edge Contraction]\label{def:edge-contracting}
		Let $G = (V, E)$ and $e = xy$. $G/e$ is the graph $G' = (V', E')$ such that $V' = (V \setminus \{x, y\}) \cup \{v_e \}$, where
		$v_e$ is a new vertex, and $E' = \{ vw \in E ~| ~\{v, w\} \cap \{x, y\} = \emptyset \} \cup 
		\{ v_e w ~| ~xw \in E \setminus \{ e \} \text{ or } yw \in E \setminus \{e\} \}$.
		
		We say that $G'$ is obtained from $G$ by \emph{contracting} the edge $e$.
	\end{defi}
	
	We use the following definition of graph minor which essentially says that a graph minor of a graph $G$ can be obtained by a series of graph contractions 
	(see~\cite{graphtheorythrded} for more details).
	
	\begin{defi}[Graph Minor]\label{def:graph-emb}
		The graph $G'$ is a \emph{graph minor} of the graph $G$, if there exist graphs
		$G_0, ~G_1, ~\ldots, ~G_n$ and edges $e_i \in G_i$ such that $G= G_0$, $G_n \simeq G'$, and $G_{i+1} = G_i/e_i$ for $i=0,\ldots,n-1$.
		We use the notation $G\succcurlyeq G'$ if $G'$ is a graph minor of $G$.
		
		Next, we extend the graph minor definition above to terms using term graphs. 
		For terms $t$ and $t'$, $t \succcurlyeq t'$ if for the 
		corresponding term graphs of $t$ and $t'$, denoted as $G$ and $G'$
		respectively, we have $G \succcurlyeq G'$.
	\end{defi}
	
	\begin{rem}
		Note, in the classical definition of graph minor (see~\cite{graphtheorythrded}), if $G$ is a subgraph of a larger graph $G_{large}$ and $G \succcurlyeq G'$, then also $G_{large} \succcurlyeq G'$. However, this component of the definition is not necessary for the results of this paper and by leaving it 
		out we are able to simplify the later definitions and presentation.   
		
		The above type of embedding, denoted by $\succcurlyeq$, provides more flexibility than the traditional subterm relation while still  preserving some features we need. 
		
	\end{rem}
	
	\begin{exa}\label{example:contractions}
		Notice that $G'$ is obtained from $G$ by first applying a sequence of edge contractions, contracting the edge depicted by $||$ at each step, resulting in $G_2$, and finally $G_2 \simeq G'$. 
		Therefore, $G \succcurlyeq G'$.

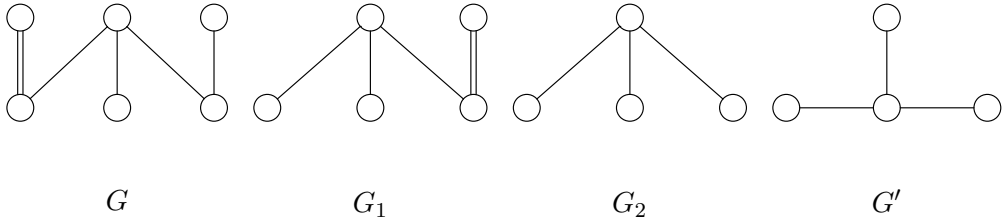
		\begin{tabular}{cccc}
			\xymatrix{
				*++=[o][F]{~} \ar@{=}[d] & *++=[o][F]{~} \ar@{-}[dl] \ar@{-}[d] \ar@{-}[dr]& *++=[o][F]{~} \ar@{-}[d] \\
				*++=[o][F]{~} & *++=[o][F]{~} & *++=[o][F]{~} \\
				& G &
			}
			&
			\xymatrix{
				& *++=[o][F]{~} \ar@{-}[dl] \ar@{-}[d] \ar@{-}[dr]& *++=[o][F]{~} \ar@{=}[d] \\
				*++=[o][F]{~} & *++=[o][F]{~} & *++=[o][F]{~} \\
				& G_1 &   
			}
			&
			\xymatrix{
				& *++=[o][F]{~} \ar@{-}[dl] \ar@{-}[d] \ar@{-}[dr]& \\
				*++=[o][F]{~} & *++=[o][F]{~} & *++=[o][F]{~}\\
				& G_2 &   
			}
			&
			\xymatrix{
				& *++=[o][F]{~} \ar@{-}[d] & \\
				*++=[o][F]{~} \ar@{-}[r] & *++=[o][F]{~} \ar@{-}[r] & *++=[o][F]{~} \\
				& G' &
			}
		\end{tabular}
		
	\end{exa}
	
	We can now extend the above graph-theoretic notions to the term rewrite setting. 
	
	\section{Graph-Embedded Systems}\label{sec:graphemb}
	The key to translating from the graph theory setting to the term setting is to use the same method, \emph{edge contractions}, but then only consider the term graphs representing \emph{well-formed} terms. That is, we need to enforce the notion of a well formed term.  
	
	To begin we need to model the graph isomorphism. A \emph{restricted} form of isomorphism can be translated into the term rewriting setting by considering permutations.
	
	\begin{defi}[Leaf and Subterm Permutations]\label{def:permutative1and2}
		We define two types of permutations, $\approx_s$ and $\approx_l$:
		\begin{enumerate}
			\item 
			For terms $t$ and $t'$, we say $t$ is \emph{subterm permutatively} equal to $t'$, denoted $t \approx_s t'$, if one of the following
			is true:
			\begin{enumerate}
				\item $t = t'$, where $t$ and $t'$ are constants or variables, or
				\item  $t= f(u_1, \ldots, u_n)$ and $t' =f(u_{\sigma(1)}, \ldots, u_{\sigma(n)})$ where $f$ is a $n$-ary function symbol, $n\geq 1$, and $\sigma$ is a permutation of the indexes $(1, \ldots, n)$. 
			\end{enumerate}
			\item 
			For terms $t$ and $t'$, we say $t$ is \emph{leaf permutatively} equal to $t'$, denoted $t \approx_l t'$, if $t'$ is obtained from $t$ via a permutation of variables and constants occurring in $t$.
		\end{enumerate}	
	\end{defi}
	The first type of permutation, $\approx_s$, allows for permutation inside the term but preserves the layer like structure of the function symbols in the term graph. The second type of permutation in the classical leaf permutability and is restricted to the leaf vertices, i.e., just the variables and constants of the term graph. 
	We will use a combination of the above two permutations in the definition employed for graph-embedded TRS. 
	\begin{defi}[Permutatively Equal]\label{def:permutative}
		For terms $t$ and $t'$, we say $t$ is \emph{permutatively} equal to $t'$, denoted $t \approx t'$, if $t \approx_s t'' \approx_l t'$, for some term $t''$.
	\end{defi}
	
	\begin{rem}
		It is useful here to remark on the motivation of the above definition, $\approx$. The goal is to model the graph isomorphism property. At the same time one needs to be careful not to be too broad and
		remove layer preserving properties of Definition~\ref{def:trs-graph-emb} and thus later protocol properties such as local stability (see Definition~\ref{def:locally-stable}). In addition, one cannot be too restricted and disallow working protocol representations such as Example~\ref{example:add} which requires more than just leaf permutability. However, it may be possible to improve upon the above definition and allow for additional systems while still maintaining the decidability of the knowledge problems shown here, see the discussion in Section~\ref{sec:conclusions}.
	\end{rem}
	
	The next step is to develop a set of rewrite \emph{schemata} which preserve a type of graph minor relation on the term graphs. This set of rewrite schemata then induces a graph-embedded term rewrite system. Notice that this is very similar to what is often done when considering the homeomorphic embeddings, see Definition~\ref{def:hom-emd}.
	
	\begin{defi}[Graph Embedding]\label{def:trs-graph-emb}
		Consider the following reduction relation, $\rightarrow^{*}_{R_{gemb}}$, where $R_{gemb}$ is the set of rules given by the instantiation of the following rule schema: 
		\[\left\{
		\begin{array}{ll}
			& \text{for any } f \in \Sigma \\
			(1) & f(x_1, \ldots, x_n) \rightarrow x_i\\ 
			(2) &f(x_1, \ldots, ~x_{i-1}, ~x_i, ~x_{i+1}~\ldots, ~x_n) \rightarrow f(x_1, \ldots, ~x_{i-1}, ~x_{i+1}, ~\ldots, ~x_n) \\
			& \text{and for any } f, g \in \Sigma \\
			(3) & f(x_1, \ldots, x_{i-1}, g(\bar{z}), x_{i+1}, ~\ldots, x_m) \rightarrow g(x_1, \ldots, ~x_{i-1}, \bar{z}, x_{i+1}, ~\ldots, x_m) \\
			(4) &f(x_1, \ldots, x_{i-1}, g(\bar{z}), x_{i+1}, ~\ldots, x_m) \rightarrow f(x_1, \ldots, ~x_{i-1}, \bar{z}, x_{i+1}, ~\ldots, x_m) 
		\end{array} 
		\right\} \]
		
		We say a term $t'$ is \emph{graph-embedded} in a term $t$, denoted $t \succcurlyeq_{gemb} t'$,  
		if $t'$ is a well formed term and there exists a term $s$ such that $t \rightarrow_{R_{gemb}}^{*} s \approx t'$. 
		
		A TRS $R$ is \emph{graph-embedded} if for any $l \rightarrow r \in R$, $l \succcurlyeq_{gemb} r$ or $r$ is a constant.
		
	\end{defi}
	
	\begin{rem}
		Notice that the rules in $R_{gemb}$ (cf. Definition~\ref{def:trs-graph-emb}) ignore function arity, thus intermediate terms between $t$ and $t'$ may not be well formed. It is only the final term for which function arity and the relation between variables and functions
		must obey the standard term definition requirements. 
		The schemata are being used to establish the graph-embedded property of a TRS.   
	\end{rem}

	\begin{rem}
		The rules in $R_{gemb}$ (cf. Definition~\ref{def:trs-graph-emb}) provide a convenient schemata for defining graph-embedded systems. However, they are also very useful in proving properties about graph-embedded systems. Notice that any rewrite step in a graph-embedded system corresponds to one or more steps of the above rules, thus proofs about graph-embedded TRSs can often be reduced to arguments on the properties of the rules in $R_{gemb}$.
	\end{rem}
        
	Definition~\ref{def:trs-graph-emb} provides a rewrite relation
	interpretation of graph-embedded systems which is contained in the $ \succcurlyeq$ relation given in Definition~\ref{def:graph-emb}. This is due to the fact that if one applies one of the rule schema from Definition~\ref{def:graph-emb} to a term, $t$, this correspond to one or more edge contractions applied to the term graph of $t$, therefore we obtain the following:

	\begin{lem}\label{lemma:contained-graph-emb}
		For any terms $t$ and $t'$, $t \succcurlyeq_{gemb} t'$ implies  $t \succcurlyeq t'$, i.e., $\succcurlyeq_{gemb} ~\subseteq ~\succcurlyeq$.
	\end{lem}

	\begin{exa}\label{example: graph-emb-ex1}
		Consider the two terms $t =f(h(a, b), h(c,d))$ and $t' = f(d, a)$.
		Then, $t \succcurlyeq_{gemb} t'$, since $t \rightarrow_{R_{gemb}}^{*} s \approx t'$ where the derivation $t \rightarrow_{R_{gemb}}^{*} s$ is as follows:
		
		\begin{tabular}{ccccc}
			
			
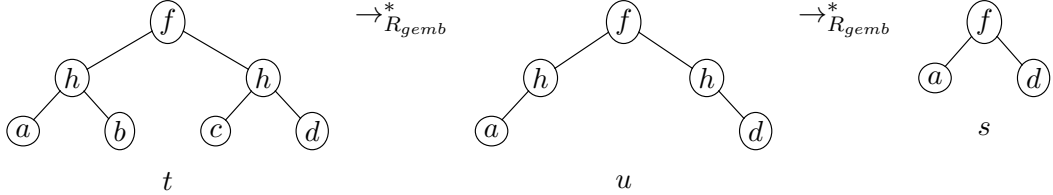
\xymatrix@-1.5pc{
				& & & *+[o][F]{f} \ar@{-}[dll] \ar@{-}[drr] & & & \\
				& *+[o][F]{h} \ar@{-}[dl] \ar@{-}[dr] & & & & *+[o][F]{h} \ar@{-}[dl] \ar@{-}[dr] & \\
				*+[o][F]{a} & & *+[o][F]{b} & & *+[o][F]{c} & & *+[o][F]{d}\\
				& & & t & & &  
			}
			& $\rightarrow_{R_{gemb}}^*$ & 
			\xymatrix@-1.5pc{
				& & & *+[o][F]{f} \ar@{-}[dll] \ar@{-}[drr] & & & \\
				& *+[o][F]{h} \ar@{-}[dl]  & & & & *+[o][F]{h} \ar@{-}[dr] & \\
				*+[o][F]{a} & &  & & & & *+[o][F]{d}\\
				& & & u & & &
			}
			& $\rightarrow_{R_{gemb}}^*$ & 
			\xymatrix@-1.5pc{
				& *+[o][F]{f} \ar@{-}[dl] \ar@{-}[dr] & \\
				*+[o][F]{a} & & *+[o][F]{d}\\
				& s &
			}
			
		\end{tabular}

	\end{exa}

	\begin{exa}[Malleable Encryption]\label{example:mal-encrypt}
		Consider the theory of Malleable Encryption, $R_{mal}$:
		\begin{align*}
			\dec(\enc(x, y), y) & \rightarrow x \\
			\mal(\enc(x,y), z) & \rightarrow \enc(z, y) 
		\end{align*}
		For the second rule, let $t_1 = \mal(\enc(x,y), z)$, and consider the following derivation $t_1 \rightarrow_{R_{gemb}}^* t_3$:
		
		\medskip
		
		\begin{tabular}{ccccc}
			
			\xymatrix@-1.5pc{
				& & & *++=[o][F]{mal} \ar@{-}[dll] \ar@{-}[dr] & \\
				& *++=[o][F]{enc} \ar@{-}[dl] \ar@{-}[dr] & & & *+[o][F]{z} \\
				*+[o][F]{x} & & *+[o][F]{y} & & \\
				& & & t_1 & 
			}
			& $\rightarrow_{R_{gemb}}^{\textit{rule (3)}}$ &
			\xymatrix@-1pc{
				& *++=[o][F]{enc} \ar@{-}[dl] \ar@{-}[dr] \ar@{-}[d] & \\
				*+[o][F]{x} & *+[o][F]{y} & *+[o][F]{z} \\
				& t_2 & 
			}
			& $\rightarrow_{R_{gemb}}^{\textit{rule (2)}}$
			\xymatrix@-1pc{
				& *++=[o][F]{enc} \ar@{-}[dl] \ar@{-}[dr] & \\
				*+[o][F]{y} & & *+[o][F]{z} \\
				& t_3 & 
			}
		\end{tabular}
		
		\medskip        
		
		Since $t_3 \approx \enc(z, y)$, we have $\mal(\enc(x, y), z) \succcurlyeq_{gemb}  \enc(z, y)$. For the first rule, $\dec(\enc(x, y), y) \succcurlyeq_{gemb} x$. Thus, $R_{mal}$ is a graph-embedded TRS. 
	\end{exa}
	
	\begin{exa}
	  The theory of blind signatures from Example~\ref{example:blind-sig} is also a graph-embedded TRS. All the rules except the final one are subterm. 
          For the final rule, 
		\[\unblind(\sign(\blind(x, y), z), y) \rightarrow_{R_{gemb}} \sign(\blind(x, y), z)\] 
		via the $R_{gemb}$ rule (1). Then, 
		\[\sign(\blind(x, y), z) \rightarrow_{R_{gemb}} \sign(x, y, z)\] 
		via the $R_{gemb}$ rule (4). Notice again that this intermediate term is not well formed. Finally 
		\[\sign(x, y, z) \rightarrow_{R_{gemb}} \sign(x, z) \approx \sign(x,z)\]
		via $R_{gemb}$ rule (2). 
	\end{exa}
	
	\begin{exa}[Addition]\label{example:add}
		Consider the theory of Addition, $R_{add}$, from~\cite{DBLP:journals/tcs/AbadiC06}:
		\begin{align*}
			\plus(x, s(y)) & \rightarrow \plus(s(x), y) \\
			\plus(x, 0) & \rightarrow x \\
			\pred(s(x)) & \rightarrow x
		\end{align*}
		$R_{add}$ is a graph-embedded TRS. Notice that $\plus(x, s(y)) \approx \plus(s(x), y)$.
	\end{exa}

	\begin{exa}[Prefix with Pairing]\label{example:prefix-pairing}
		The theory of prefix with pairing~\cite{cortier2006survey, Dreier2017} is a graph-embedded TRS:
		\begin{align*}
			\dec(\enc(x,y), y) &\rightarrow x \\
			\prefix(\enc(<x, y>, z)) & \rightarrow \enc(x, z) \\
			\fst(<x, y>) & \rightarrow x \\
			\snd(<x, y>) & \rightarrow y
		\end{align*} 
	\end{exa}

	\begin{exa}[Trap-door Commitment]\label{example:trap-door-commit}
		The theory of trap-door commitment~\cite{DBLP:journals/jar/CiobacaDK12} is a graph-embedded TRS:
		\begin{align*}
			\open(\td(x, y, z), y) &\rightarrow x \\
			\open(\td(x_1,y,z), f(x_1, y, z, x_2)) &\rightarrow x_2 \\
			\td(x_2, f(x_1, y, z, x_2), z) &\rightarrow \td(x_1, y, z)\\
			f(x_2, f(x_1, y, z, x_2), z, x_3) &\rightarrow f(x_1, y, z, x_3) \\
		\end{align*}

	\end{exa}

        \begin{exa}
         \label{example:strong-conf}
                  A subterm TRS $R$ is graph-embedded since for any rule $l \rightarrow r \in R$, $r$ is either a constant or a subterm of $l$ obtained  from $l$ by the repeated application at the root position of  $R_{gemb}$ rule (1). For example, the following TRS specifying a form of strong secrecy~\cite{blanchet2004automatic, DBLP:journals/tocl/ChadhaCCK16} is subterm convergent and so graph-embedded:
		\begin{align*}
	 		\fst(<x,y>) &\rightarrow x \\
	 		\snd(<x,y>) &\rightarrow y \\	
	 		\adec(\aenc(x, \pk(y)), y) &\rightarrow x \\
	 		\dec(\enc(x, y), y) &\rightarrow x \\
	 		\chck(\sign(x,y), \pk(y)) &\rightarrow \ok\\
	 		\msg(\sign(x,y)) &\rightarrow x \\
	 	\end{align*}                  
        \end{exa}

	\subsection{Some Properties of Graph-Embedded Systems}
	As an initial step we explore some of the basic properties of the graph-embedded TRSs. Similar to the class of subterm TRSs, the graph-embedded TRSs have several nice properties such as termination. We can first note that the $\succcurlyeq_{gemb}$ relation is a partial order on the class of terms. This follows from Lemma~\ref{lemma:contained-graph-emb} and the fact that the graph-embedded relation is a partial ordering on the class of finite graphs (See Proposition 1.7.3 from~\cite{graphtheorythrded}).
	
        Graph-embedded systems have the nice property of being size reducing when rewrite steps are applied and thus terminating. 

	\begin{lem}\label{lemma:graph-emb-reduce}
	  Let $R$ be a graph-embedded TRS.
For any terms $t,t'$ such that $t \rightarrow_{R} t'$, we have that  
for any function symbol or variable $\mathit{fv}$, $|t|_{\mathit{fv}} \geq |t'|_{\mathit{fv}}$. 
If for any $l \rightarrow r \in R$, $l \rightarrow_{R_{gemb}} ^{+} \cdot \approx r$, then $t \rightarrow_{R} t'$ implies $|t| > |t'|$, and so $R$ is terminating.
        \end{lem}  

        \begin{proof}
        It is easy to check that no rule from Definition~\ref{def:trs-graph-emb} introduces additional function symbol or variable occurrence. All rules from Definition~\ref{def:trs-graph-emb} remove function symbols except the second rule and $\approx$. Notice that if rule (2) is applied then one of the other rules must also be applied to ensure the final term is well formed. Finally, if we require that at least one rewrite step is applied, then the size of the term is strictly reduced even if $\approx$ does not reduce the size of the term. Thus, it corresponds to a strictly decreasing measure, showing that $R$ is terminating.
	\end{proof}

	Let us compare the definitions of graph embedded and homeomorphic embedding.

        \begin{lem}
          The class of homeomorphic-embedded TRSs is strictly included in the class of graph-embedded TRSs.
        \end{lem}

        \begin{proof}
        Applying only $R_{gemb}$ rule (1), we see that homeomorphic-embedded TRSs are a subset of graph-embedded TRSs. To show that this is a strict subset,    
        consider Malleable Encryption, $R_{mal}$, from Example~\ref{example:mal-encrypt}. 
	$R_{mal}$ is a graph-embedded TRS, as is shown in Example~\ref{example:mal-encrypt}. However, 
	$R_{mal}$ is not a homeomorphic-embedded TRS. This can be seen in the rule 
	$\mal(\enc(x,y),z) \rightarrow \enc(z,y)$. There is no way to obtain the term $\enc(z,y)$ from the term $\mal(\enc(x,y),z)$ by applying only $R_{gemb}$ rule (1).
        \end{proof}

In the rest of the paper we look at how graph-embedded TRSs can be
used to both extend results in security protocols and also give a
formal syntactic definition to classes of protocol presentations for
which the decidability of the two knowledge problems are already
known but are not contained in the class of strictly subterm convergent theories.  We focus on theories with the local stability property as
introduced in~\cite{DBLP:journals/tcs/AbadiC06} (and extended
in~\cite{DBLP:journals/tcs/Ayala-RinconFN17}).  For this purpose, we
need to consider a restricted form of graph-embedded system called
contracting system introduced in
Definition~\ref{def:cap-contracted}. One can show that without such a
restriction, the knowledge problems for graph-embedded TRSs are
undecidable in general.
	
	\section{Undecidable Knowledge Problems}
        \label{sec:undecidable-know}
	It is shown in~\cite{DBLP:journals/tcs/AbadiC06} that the knowledge problems are undecidable in general. For graph-embedded systems we construct a proof that uses a reduction from the Modified Post Correspondence Problem (MPCP), a well known undecidable problem. This proof is similar to two other undecidability proofs developed in~\cite{Anantharaman2012JAR}, and in the research report~\cite{Anantharaman2007HAL} of the paper~\cite{DBLP:conf/rta/AnantharamanNR07}. 
	
	Let $\Gamma = \{a, b\}$ be the alphabet of the $MPCP$ problem. Then, an instance of the problem is a finite set of string pairs, 
	$S = \{(\alpha_i, \beta_i) ~|~ i \in [1, n], \alpha_i \in \Gamma^{+}, \beta_i \in \Gamma^{+} \}$, and two input strings $\alpha_0 \in \Gamma^{*}$ , and $\beta_0 \in \Gamma^{*}$. A solution is a sequence of indexes 
	$i_1, \ldots i_k \in [1,n]$ such that 
	$\alpha_{i_1} \alpha_{i_2} \ldots \alpha_{i_k} = \beta_{i_1} \beta_{i_2} \ldots \beta_{i_k}$, $\alpha_0$ is a suffix of $\alpha_{i_1} \alpha_{i_2} \ldots \alpha_{i_k}$, and $\beta_0$ is a suffix of $\beta_{i_1} \beta_{i_2} \ldots \beta_{i_k}$. That is, there are string $\alpha', \beta' \in \Gamma^*$ s.t.
	$\alpha_{i_1} \alpha_{i_2} \ldots \alpha_{i_k} = \alpha' \cdot \alpha_0 = \beta' \cdot \beta_0 = \beta_{i_1} \beta_{i_2} \ldots \beta_{i_k}$. Notice that the standard $PCP$ is easily reducible to this $MPCP$ by setting
	$\alpha_0$ and $\beta_0$ to the empty string. 
	
	\begin{lem}\label{lemma:graph-embedded-undecidable}
		The deduction problem is undecidable for the class of homeomorphic-embedded convergent TRSs. 
	\end{lem}
	\begin{proof}
		Consider the $MPCP$ defined by $\alpha_0, ~\beta_0$, and the finite set of string pairs $\{(\alpha_i, \beta_i) | ~i \in [1, n]\}$ over the alphabet $\Gamma = \{a, b\}$. Assume unary function symbols $locked$, $unlocked$, $a$, $b$, $g_1,\dots,g_n$, a quaternary function $f$, and constants $c$, $d$, $e$. 
		
		Each string from the $MPCP$ can be viewed as a sequence of applications of the unary function symbols. We can convert a string over $\Gamma$ into a term using function symbols $a$ and $b$ as follows. For an alphabet symbol $\widetilde{\gamma} \in \Gamma$, represent by $\gamma$ the corresponding function symbol. That is, if $\widetilde{\gamma} = a$, then $\gamma = a()$. Likewise,  if $\widetilde{\gamma} = b$, then $\gamma = b()$. Now for each pair of strings, $(\alpha_i, \beta_i)$, we can recursively construct the function interpretation of each string, $\widetilde{\alpha}$, as follows: $\widetilde{\alpha_i}(x) = \widetilde{\gamma \alpha_i'}(x) = \gamma(\widetilde{\alpha_i'}(x))$, and $\widetilde{\beta_i}(x) = \widetilde{\gamma \beta_i'}(x) = \gamma(\widetilde{\beta_i'}(x))$.

		Let $R = B \cup U$ where 
			\begin{align*}
					B & = \bigcup_{i=1}^n \{ f(\widetilde{\alpha_i}(x), g_i(y), \widetilde{\beta_i}(z), unlocked(k)) \rightarrow f(x, y, z, unlocked(k)) \}, \\ 
					U & = \{ f(x, y, x, locked(unlocked(z))) \rightarrow f(x, y, x, unlocked(z)) \}.
			\end{align*}
                        The rules in $B$ are called the \emph{block rules}, while the rule in $U$ is called the \emph{unlock} one.
			The function symbols $g_1,\dots,g_n$ (and the $locked$ in the final rule) ensure there are no critical pairs between rules and thus we have a convergent TRS.

		Now construct the frame $\phi = \nu \tilde{n}.\sigma$ where $\tilde{n} = \{c, e\}$ and $\sigma = \{ x \mapsto \widetilde{\alpha_0}(c), y \mapsto \widetilde{\beta_0}(c), z \mapsto locked(unlocked(e))\}$, and let the target ground term to be deduced be $f(c,d,c,unlocked(e))$.

		Let us show that there exists a recipe $\zeta$ such that $\zeta\sigma \rightarrow_{R}^{!} f(c, d, c, unlocked(e)) $  iff there is a solution to the $MPCP$.
		\begin{itemize}
			\item Let $\alpha_{i_1} \alpha_{i_2} \ldots \alpha_{i_k} = \beta_{i_1} \beta_{i_2} \ldots \beta_{i_k}$, be a solution to the $MPCP$ problem.
			Let $\alpha' \cdot \alpha_0 = \alpha_{i_1} \alpha_{i_2} \ldots \alpha_{i_k}$ and $\beta' \cdot \beta_0 = \beta_{i_1} \beta_{i_2} \ldots \beta_{i_k}$. Then, $\widetilde{\alpha'}(x)\sigma = \widetilde{\beta'}(y)\sigma$. Now, for each block, $i$ in the $MPCP$, there exists a single block rule, containing $g_i$, that removes the $\alpha_i$ and $\beta_i$.
			Therefore, $\zeta = f(\widetilde{\alpha'}(x), g_{i_1}(\ldots (g_{i_k}(d))), \widetilde{\beta'}(y), z)$ is a solution recipe since:
			\begin{align*}
				\zeta \sigma  &=   f(\widetilde{\alpha'(\alpha_0(c))}, g_{i_1}(\ldots (g_{i_k}(d))), \widetilde{\beta'(\beta_0(c))}, locked(unlocked(e)))  \rightarrow_R^{unlock} \\
				& f(\widetilde{\alpha'(\alpha_0(c))}, g_{i_1}(\ldots (g_{i_k}(d))), \widetilde{\beta'(\beta_0(c))}, unlocked(e)) \rightarrow^{!}_{R} f(c, d, c, unlocked(e))
			\end{align*}
			
			\item Let $\zeta$ be a recipe term such that  $\zeta \sigma \rightarrow^{!}_{R} f(c,d,c, unlocked(e))$ then there is a solution to the $MPCP$ and it can be extracted from the 
			indexes of the $g_i$ function symbol in the term $\zeta \sigma$. This is due to the fact that the two strings must be the same, otherwise the unlock rule in $U$ cannot be applied. Notice also that this rule must be applied before any of the remaining rules since it ``unlocks'' the function $unlocked$ which is required for the other block rules. Next since each of the rules represents a string pair in the $MPCP$ problem, if $\zeta \sigma$ reduces to the target $f(c, d, c, unlocked(e))$, then the strings in $\zeta \sigma$ are composed exactly of $MPCP$ string pairs.
			Finally, the recipe cannot just be $f(c,d,c, unlocked(e))$ since 
			$c, e \in \tilde{n}$. In addition, the recipe cannot contain $unlocked(e)$ before applying the substitution $\sigma$, since $e \in \tilde{n}$.  \qedhere
		\end{itemize} 
	\end{proof}
	
	\begin{exa}\label{example:pcp}
		Consider the following $MPCP$:
		\[ \overbrace{\left( \frac{ba}{baa} \right) }^\text{\text{pair 1}}
		,~\overbrace{ \left( \frac{ab}{ba} \right) }^\text{\text{pair 2}}
		,~\overbrace{ \left( \frac{\vphantom{b}aaa}{aa} \right) }^\text{\text{pair 3}} \]
		Assume the two starting string are: $\alpha_0 = aa$ and $\beta_0 = a$.
		
		Following the construction of Lemma~\ref{lemma:graph-embedded-undecidable}:
		\[R = \left\{
		\begin{array}{ll}
			f(b(a(x)), g_1(y), b(a(a(z))), unlocked(k)) & \rightarrow f(x, y, z, unlocked(k)) \\ 
			f(a(b(x)), g_2(y), b(a(z)), unlocked(k)) & \rightarrow f(x, y, z, unlocked(k)) \\ 
			f(a(a(a(x))), g_3(y), a(a(z)), unlocked(k)) & \rightarrow f(x, y, z, unlocked(k)) \\
			f(x, y, x, locked(unlocked(z))) & \rightarrow f(x, y, x, unlocked(z)) 
		\end{array} 
		\right\} \]
		In addition, $\widetilde{\alpha_0}(c) = a(a(c))$ and $\widetilde{\beta_0}(c) = a(c)$.
		Notice there is a solution with indexes 1 and 3. That is, 
		$baa \cdot \alpha_0 = \alpha_1 \cdot \alpha_3  =  baaa \cdot \beta_0 = \beta_1 \cdot \beta_3 $.
		
		\noindent
		The starting frame is: $\phi = \nu \{c, e\}.\{  
		x \mapsto a(a(c)), ~y\mapsto  a(c), z \mapsto locked(unlocked(e)) \}$.
		Then, there is a solution recipe: 
		$\zeta = f(b(a(a(x))), g_1(g_3(d)), b(a(a(a(y)))), z)$, since $\zeta \sigma \rightarrow^{!}_R f(c, d, c, unlocked(e))$.
	\end{exa}

	\noindent
	As a corollary of Lemma~\ref{lemma:graph-embedded-undecidable} we obtain the following.
	
	\begin{cor}\label{cor:undec-deduction}
		The deduction problem is undecidable for the class of graph-embedded convergent TRSs.
	\end{cor}

	\begin{rem}
		Note that the knowledge problems of deduction and static-equivalence were already proven undecidable in general in~\cite{DBLP:journals/tcs/AbadiC06}, where a reduction from PCP is also used. However, the system used in the proof from~\cite{DBLP:journals/tcs/AbadiC06} is not graph-embedded and it is not clear how to directly adapt that proof to the graph-embedded case of Lemma~\ref{lemma:graph-embedded-undecidable}.
		
		Finally, let us note that the above proof fixes an error with the proof of the same result presented in~\cite{DBLP:conf/fscd/SatterfieldEMR23}. There, the system would erroneously allow trivial solutions since it did not include an unlock type rule that forces the two string to be equal. 
	\end{rem}
	
	\section{Decidable Knowledge Problems}
	\label{sec:decidable-know}
        
	We consider below a restricted form of the graph-embedded TRS for which we
	can show decidability of the knowledge problems. The key to the definition is to start with a
	graph-embedded system but then to require a set of ``projection rules''. These rules, in a way, 
	provide access to subterms which may be modified during some rewrite derivation. Access to these subterms
	is important for the success of the saturation based decision procedure \cite{DBLP:journals/tcs/AbadiC06}
	for showing decidability of the knowledge problems (see Definition~\ref{def:sat} below). Having access to subterms is a critical property in a number of definitions for which the knowledge problems are considered. It is a key feature in subterm convergent TRSs. It is also a feature in $\Delta$-strong theories, defined in~\cite{DBLP:conf/rta/AnantharamanNR07}, where the cap problem, a form of the deduction problem, is studied for such theories (see Section~\ref{sec:relations-cap}). We can now introduce the type of projection rules we will require. 

        \begin{defi}[Projecting Rule]\label{def:projecting}
        Let $R$ be a TRS, $x$ a variable and $t$ a term with a single occurrence of $x$. A {\em projecting rule in $R$ over $t$ leading to $x$} is a rule in $R$ which is a variant of $t' \rightarrow x$ where $t'$ is a superterm of $t$ with no additional occurrences of $x$.
        \end{defi} 
        
        \begin{rem}
        	The requirement that $x$ only occurs once in the left-hand side allows us to rule out projecting rules which do not help in producing new knowledge during the saturation procedure (Definition~\ref{def:sat}). For example, a rule of the form $f(x, g(x)) \rightarrow x$ requires that you already have knowledge of the subterm of $g()$ in order to derive the subterm. However, this restriction does not require that the entire term be linear. For example, $\dec(\enc(x, y), y) \rightarrow x$ is a valid projecting rule since you only need to know $y$ to unlock the $enc()$ term and obtain the subterm instantiating $x$.
        \end{rem}
        
        How these projecting rules relate to the entire term-rewrite system can now be defined. We consider the relation to a rewrite derivation, using the graph-embedded rules, and then the permutation relation. 
        
        \begin{defi}[Projection-Closed Derivation]\label{def:projection-closed-deriv}
          Let $R$ be a graph-embedded TRS.
          A non-empty $R_{gemb}$-derivation is said to be {\em projection-closed with respect to $R$}   if it has the form
          $$s \rightarrow_{R_{gemb}}^{l\rightarrow r,\gamma} s' \rightarrow_{R_{gemb}}^* t$$
          where
          \begin{itemize}
          \item $s$ is a linear term and $t$ is a well formed term,  
          \item $l \rightarrow r$ is a $R_{gemb}$ rule among (1), (2) and (4),
          \item $\gamma$ is a substitution ranging over variables,
          \item if $l\rightarrow r$ is rule (1), $f(x_1, \ldots, x_n) \rightarrow x_i$, then for any $x_i\gamma$ occurring in $t$ there exists a projecting rule in $R$ over $f(x_1, \ldots, x_n)\gamma$ leading to $x_i\gamma$,
          \item if $l\rightarrow r$ is rule (4), $f(x_1, \ldots, x_{i-1}, g(\bar{z}), x_{i+1}, ~\ldots, x_m) \rightarrow f(x_1, \ldots, ~x_{i-1}, \bar{z}, x_{i+1}, ~\ldots, x_m)$, then for any $z_i\gamma$ occurring in $t$ there  exists a projecting rule in $R$ over $g(\bar{z})\gamma$ leading to $z_i\gamma$,
          \item $s' \rightarrow_{R_{gemb}}^* t$ is either projection-closed with respect to $R$ or empty.
          \end{itemize}
          \end{defi}

          \begin{defi}[Projection-Closed Permutative Equality]\label{def:projection-closed-permut-eq}
            Let $R$ be a graph-embedded TRS and $l = r$ a permutative equality obtained by one or zero applications of $\approx_s$ followed by one or zero applications of $\approx_l$. The equality $l = r$ is said to be {\em projection-closed with respect to $R$} if the following holds: for any variable $x$ occurring in a strict subterm $l' \neq x$ of $l$ and in a strict subterm $r'$ of $r$ such that $l' \neq r'$, then there exists a projecting rule in $R$ over $l'$ leading to $x$.
        \end{defi}

        Based on Definitions~\ref{def:projection-closed-deriv} and~\ref{def:projection-closed-permut-eq},  we can now introduce the subset of graph-embedded rules for which we can obtain decidability of the knowledge problems.

          \begin{defi}[Contracting TRS]\label{def:cap-contracted}
            
          A graph-embedded TRS $R$ is said to be {\em contracting} if for any non-subterm rule $l\rightarrow r$ in $R$ one of the following holds:
        \begin{itemize}  
        \item there exist a position $p$ of $l$, a substitution $\sigma$ ranging over variables, a projection-closed derivation $g \rightarrow_{R_{gemb}}^+ d$ with respect to $R$ such that $l{|_p} = g\sigma$ and $r=d\sigma$ is of depth $1$,
        \item $l = r$ is a permutative equality such that $depth(l)=depth(r)=2$ and $l=r$ is projection-closed with respect to $R$.
        \end{itemize}  

        A contracting TRS, $R$, is \emph{strictly contracting} if for any $l \rightarrow r \in R$, $depth(l)>depth(r)$.

        \end{defi}

    \begin{rem}
     Note that the projection-closed derivation in Definition~\ref{def:cap-contracted} starts at some subterm of $l$, say $l|_p$. To get this subterm $l|_p$ from $l$ thanks to a rewriting derivation with respect to $R_{gemb}$, it suffices to apply repeatedly $R_{gemb}$ rule (1) at the root until it reaches $l|_p$. These root applications of rule (1) do not necessarily require projecting rules, as the non-root applications of rule (1) will in the projection-closed derivation. See Example~\ref{example:cap-contracting} below for an illustration of this fact.
    	
     It is also worth pointing out that the requirement on $g$ being a linear term (in Definition~\ref{def:cap-contracted}) is a way to identify individual instances of variables, which is needed to identifying the variable for which a projection rule is required (in Definition~\ref{def:projection-closed-deriv}). However, the requirement that $g$ be linear is not a restriction requiring that $l$ (in Definition~\ref{def:cap-contracted}) be a linear term. Notice that 
     $\sigma$ is a substitution ranging over variables and used to match $l|_p$ with $g$. Thus, even if $g$ is linear, it does not require that $l$ is linear.  
    \end{rem}
             	
	Although the definition restricts the set of graph-embedded systems it is
	still sufficient to model many security protocols of interest. We include several such examples below.

	\begin{exa}\label{example:cap-contracting}
		Consider the theory of blind signatures from Example~\ref{example:blind-sig}. 
		This theory is contracting. Let us look at the rule: 
		$\unblind(\sign(\blind(x, y), z), y) \rightarrow \sign(x, z)$ and see how it satisfies the
		above definitions for contracting convergent systems by identifying the subterm $l|_{p}$, and the projection-closed derivation. Let us construct this derivation step by step.
		\begin{itemize}
			\item Let $l = \unblind(\sign(\blind(x, y), z), y)$. We start the projection-closed derivation at the subterm $\sign(\blind(x, y), z)$, that is $l|_{1} = \sign(\blind(x, y), z)$. 
			\item There exists a projection-closed derivation starting from $l|_{1}$, considering the substitution $\gamma = \{ x_1 \mapsto x, ~x_2 \mapsto y, ~x_3 \mapsto z \}$:
                          $$\sign(\blind(x, y), z) \rightarrow_{R_{gemb}}^{\textit{rule (4)},\gamma} \sign(x, y, z) \rightarrow_{R_{gemb}}^{\textit{rule (2)},\gamma} \sign(x, z)$$
where $\sign(x, z)$ is of depth $1$. 
		\end{itemize}		
In the above derivation, the first step $\sign(\blind(x, y), z)  \rightarrow_{R_{gemb}}^{\textit{rule (4)},\gamma} \sign(x, y, z)$ requires a projecting rule. There are two variables, $x$ and $y$, under the $blind()$ symbol. However, notice that $y$ is removed in the next step. Therefore, we need only one projecting rule, $t' \rightarrow x$, which is the rule $\unblind(\blind(x, y), y) \rightarrow x$. Following Definition~\ref{def:projecting}, $t' = \unblind(\blind(x, y), y)$ is a superterm of $t = \blind(x, y)$ where $x$ only occurs once in $t'$.

		Notice that the remaining rules of the blind signatures theory are subterm. Therefore, blind signatures is a contracting convergent TRS.
	\end{exa}
	
		\begin{exa}\label{example:cap-contracting2}
		Consider several additional convergent TRSs given in previous examples:
		\begin{itemize}
			\item The theory of addition, introduced in Example~\ref{example:add}, is a contracting TRS and provides a good example of the use of permutative equality in Definition~\ref{def:cap-contracted}.
			Consider the rule $\plus(x, s(y)) \rightarrow \plus(s(x), y) \in R_{add}$.  
			Notice that $\plus(x, s(y)) \approx \plus(s(x), y)$ and that on the left-hand side $y$ 
			appears in the subterm $s(y)$ but does not appear in the same subterm on the right-hand side. Thus, there needs to be the rule $\pred(s(x))  \rightarrow x$.
			
			\item The theory of prefix with pairing from Example~\ref{example:prefix-pairing} is a contracting TRS.
			
			\item Any subterm convergent TRS is contracting. For example, the theory of pairing with encryption, 
			$R =\{ \fst(\left\langle x, y \right\rangle ) \rightarrow x, 
			\snd(\left\langle x, y \right\rangle ) \rightarrow y,
			\dec(\enc(x,y),y) \rightarrow x\}$, 
			and the theory of Example~\ref{example:strong-conf}.

		\end{itemize}
	\end{exa}

	\begin{exa}\label{example:not-cap-contracting}
		Consider the following TRSs:
		\begin{itemize}
			\item The theory of trap-door commitment of Example~\ref{example:trap-door-commit} is not a contracting TRS. Notice that it is missing the required projecting rules. Indeed, for the rule $td(x_2, f(x_1, y, z, x_2), z) \rightarrow td(x_1, y, z)$, there would need to be projecting rules for $x_1$, $y$, and $z$.
			Interestingly, this theory is also not locally stable~\cite{DBLP:journals/jar/CiobacaDK12}. However, if we add the rules, \[ \fst(f(x_1, x_2, x_3, x_4)) \rightarrow x_1,
			\snd(f(x_1, x_2, x_3, x_4)) \rightarrow x_2,
			\thd(f(x_1, x_2, x_3, x_4)) \rightarrow x_3, \] then the theory is contracting and so locally stable by Theorem~\ref{thm:locally-stable}.
			
			\item The theory of Example~\ref{example:mal-encrypt} is not contracting. Notice that for the rule $\mal(\enc(x, y), z) \rightarrow \enc(z, y)$, the
			vertex labeled with $z$ is moved under the $\enc$ vertex on the right-hand side. This violates
			the requirements of Definition~\ref{def:cap-contracted}, which does not allow the $R_{gemb}$ rule (3), which would be required in this case. Thus, even with additional projecting
			rules, it cannot be made to be contracting. This theory is also
			not locally stable, as shown in~\cite{DBLP:journals/jar/CiobacaDK12}.
			
		\end{itemize}
	\end{exa}

	\begin{rem}
		If we consider now the TRS from the undecidability proof of Lemma~\ref{lemma:graph-embedded-undecidable} we can see that while the system is graph-embedded it is not contracting. The required projecting rules do not exist.
	\end{rem}
	
	We now develop a few results and definitions we need to show the 
	decidability of the knowledge problems for contracting, graph-embedded 
	convergent systems.

	\begin{defi}[Graph-Embedded Subterms]\label{def:ex-subset}
	 Let $R$ be a contracting TRS and $\phi = \nu \tilde{n}.\sigma$ be a frame.  The set of graph-embedded subterms of a term $t$, denoted as $gst(t)$, is defined as:
		$gst(c) = \{c \}$, where $c$ is a name or a constant, $gst(t) = \{ t' | t \rightarrow_{R_{gemb}}^{*} t'' \approx t', \text{ and } t' \text{ is a well formed term }\}
		\cup \bigcup_{t'' \in st(t)} gst(t'')$. In addition, $gst(\phi) = \cup_{t \in Ran(\sigma)} gst(t)$.
	\end{defi}
	
	Notice that for any term $t$, $gst(t)$ is a finite set. This is due to the fact that when recursively constructing
	$gst(t)$ in Definition~\ref{def:ex-subset}, $t'$ is equal or smaller in size to $t$, and any term 
	$t'' \in st(t)$ must be strictly smaller than $t$. Thus, we have the following result. 
	\begin{lem}\label{lemma:gst-finite}
		For any term $t$ and any frame $\phi$, $gst(t)$ and $gst(\phi)$ are finite sets.
	\end{lem}
	
	Based on the extended definition of subterms, $gst$, we can now construct a saturation set for frames. Computing such a saturation set is the goal of many procedures that consider security notions such as deducibility. The saturation set represents the knowledge of the attacker and their ability to deduce a term from that knowledge, see~\cite{DBLP:journals/tcs/AbadiC06} for more background. When considering this saturation procedure, recall that for any frame, $\phi = \nu \tilde{n}.\sigma$, $\sigma$ is a ground substitution. For the following it is also useful to recall the definition of $c_{R}$ from Section~\ref{sec:prelim}.
	
	\begin{defi}[Frame Saturation for Contracting Convergent TRSs]\label{def:sat}
          Let $R$ be any contracting convergent TRS and $\phi = \nu \tilde{n}.\sigma$ any $R$-normalized frame.
          Define the set $sat(\phi)$ to be the smallest set such that $Ran(\sigma) \cup \fn(\phi) \subseteq sat(\phi)$, and closed under the following two rules:
		\begin{enumerate}
			\item if $M_1, \ldots, M_l \in sat(\phi)$ and $f(M_1, \ldots, M_l) \in gst(\phi)$, then 
			$f(M_1, \ldots, M_l) \in sat(\phi)$,
			\item if $M_1, \ldots, M_l \in sat(\phi)$, $C[M_1, \ldots, M_l] \rightarrow_{R}^{\epsilon} M$, where $C$ is a context, $|C| \leq c_R$, $\fn(C) \cap \tilde{n} = \emptyset$, and $M \in gst(\phi)$, then $M \in sat(\phi)$.
			
		\end{enumerate}
	\end{defi}

	The finiteness of $sat(\phi)$ is critical to computing the possible attackers knowledge thus having a finite set is useful for any practical procedure for deciding deducibility. 
	
	\begin{lem}\label{lemma:sat-finite}
          Let $R$ be any contracting convergent TRS. 
          For any $R$-normalized frame $\phi$, $sat(\phi)$ is a finite set of terms deducible from $\phi$ modulo $R$.
	\end{lem}
	\begin{proof}
		New terms not originally contained in the finite set $Ran(\phi) \cup \fn(\phi)$ are only added to $sat(\phi)$ if they are first contained in 
		$gst(\phi)$. Since $gst(\phi)$ is finite by Lemma~\ref{lemma:gst-finite}, $sat(\phi)$ is finite. Moreover, $Ran(\phi) \cup \fn(\phi)$ is a set of terms deducible from $\phi$ modulo $R$. Then, the terms that are added by the two rules given in Definition~\ref{def:sat} are also deducible from $\phi$, either by congruence (cf. first rule) or by rewriting modulo $R$ (cf. second rule).
	\end{proof}

	\begin{rem}
		In Definition~\ref{def:sat}, it would be tempting
		to just place all of $gst(\phi)$ into $sat(\phi)$ immediately, but this
		would add non-deducible terms to the set and invalidate the results. Thus, we include in $sat(\phi)$ only terms belonging to $gst(\phi)$ that are also deducible from $\phi$. 
	\end{rem}

	\subsection{Closure Under Small Context}
        
	The following definition and lemmas will be useful in proving the main motivating result as they show key components of the local stability property given in Definition~\ref{def:locally-stable}.
	
	\begin{defi}[Closure Under Small Context]\label{def:closed-under-context}
                    Let $R$ be any contracting convergent TRS and $\phi = \nu \tilde{n}.\sigma$ any $R$-normalized frame. A finite set $\mathcal{S}$ of ground terms deducible from $\phi$ modulo $R$ is \emph{closed under small $\phi$-restricted context by} $R$ if the following property holds: 
		for any context $C$ with $|C| \leq c_{R}$ and 
		$\fn(C) \cap \tilde{n} = \emptyset$, and any $S_1,\ldots, S_l \in \mathcal{S}$, if $C[S_1, \ldots, S_l] \rightarrow_{R}^{\epsilon} M$ then there exist a context $C'$ and $S_1',\ldots, S_k' \in \mathcal{S}$ such that $|C'| \leq c_{R}^2$, $\fn(C') \cap \tilde{n} = \emptyset$, and $M \rightarrow_{R}^* C'[S_1', \ldots, S_k']$. When $\phi$ is clear from the context, $\mathcal{S}$ is said to be \emph{closed under small context by} $R$.
	\end{defi}
	
	Using $c_R^2$ as an upper bound is somewhat arbitrary since we need just some fixed bound. We use $c_R^2$ since it is sufficient for the results
	in this paper and it is the bound used in~\cite{DBLP:journals/tcs/AbadiC06}.

       In the following, we show that the frame saturation for a contracting convergent TRS $R$ (cf. Definition~\ref{def:sat}) is closed under small context by $R$. 
	
	\begin{lem}\label{lemma:small-context}
		For any contracting convergent TRS $R$ and any $R$-normalized frame $\phi$, let $sat(\phi)$ be the set given in Definition~\ref{def:sat}. Then, $sat(\phi)$ is closed under small context by $R$.
	\end{lem}
	
\begin{proof}
          Let $t=C[S_1,\dots,S_l]$ be a term such that $S_1,\ldots,S_l \in sat(\phi)$, $\fn(C) \cap \tilde{n} = \emptyset$, and $|C| \leq c_{R}$.
          Assume $t \rightarrow_{R}^{\epsilon} t'$, meaning that there exist a rule $l\rightarrow r \in R$ and a substitution $\gamma$ such that $t=l\gamma$ and $t'=r\gamma$.  By analyzing the three different forms of rule $l\rightarrow r$ that occurs in a contracting TRS $R$, we show that $t'$ is necessarily a restricted context instantiated by terms in $sat(\phi)$ such that the size of the context part of $t'$ is less than or equal to $|C|$.

                \begin{enumerate}
		  
		\item Consider $l \rightarrow r$ is a subterm rule, then $r\gamma$ is a subterm of $l\gamma$ and only two cases can occur:
                \begin{itemize}
                \item $r\gamma$ is a restricted context instantiated by terms in $sat(\phi)$ whose context part is a subterm of $C$;
                \item $r\gamma$ is either a constant or a subterm of some term in $sat(\phi)$. By Definition~\ref{def:sat}, $r\gamma$ is in $sat(\phi)$.  
		\end{itemize}
                \noindent In both cases, $r\gamma$ is a restricted context instantiated by terms in $sat(\phi)$. Moreover, the size of the context part of $r\gamma$ cannot be greater than $|C|$.

		\item Consider $l \rightarrow r$ is a rule such that $depth(l)> depth(r)=1$. Assume $r = h(x_1,\dots,x_n)$ and let $i$ be any integer in $\{1,\dots, n\}$. Any $x_i$ occurs in $l$ either at depth at most $1$ or at some depth strictly greater than $1$.
		
		\begin{itemize}
			\item If $x_i$ occurs in $l$ at depth at most $1$, then $x_i\gamma$ is a restricted context instantiated by terms in $sat(\phi)$.
			\item if $x_i$ occurs in $l$ at some depth strictly greater than $1$, then $x_i\gamma$ is in $sat(\phi)$ thanks to Definition~\ref{def:projection-closed-deriv}.
		\end{itemize}
                
                \noindent Then, $r\gamma = h(x_1\gamma,\dots,x_n\gamma)$ is a restricted context instantiated by terms in $sat(\phi)$. Given the context part $C_i$ of $x_i\gamma$ for any $i=1,\dots,n$, we have that $\sum_{i=1}^n |C_i|$ is strictly less than $|C|$. Therefore, the size of the context part of $r\gamma$ cannot be greater than $|C|$.

		\item Consider $l \rightarrow r$ is a rule such that $depth(l)= depth(r) = 2$. Assume $r=h(r_1,\dots,r_n)$, and let $i$ be any integer in $\{1,\dots, n\}$.
		
		\begin{itemize}
			\item If $r_i$ is a variable occurring at depth at most $1$ in $l$, then $r_i\gamma$ is a restricted context instantiated by terms in $sat(\phi)$.
			
			\item If $r_i$ is a variable occurring at depth $2$ in $l$, then $r_i\gamma$ is in $sat(\phi)$ thanks to Definition~\ref{def:projection-closed-permut-eq}.
			
			\item If $r_i$ is a non-variable term occurring as a direct subterm $l_j$ of $l$ for some $j \in [1,n]$, then $r_i\gamma = l_j\gamma$ is a restricted context instantiated by terms in $sat(\phi)$.
			
			\item If $r_i$ is a non-variable term $f(\bar{x})$ not occurring as a direct subterm of $l$, then there are two cases for any variable $x \in \bar{x}$: if $x$ occurs at depth at most $1$ in $l$, then $x\gamma$ is a restricted context instantiated by terms in $sat(\phi)$; otherwise $x$ also occurs at depth $2$ in $l$ and $x\gamma$ is in $sat(\phi)$ thanks to Definition~\ref{def:projection-closed-permut-eq}.
		\end{itemize}

                \noindent Then, $r\gamma = h(r_1\gamma,\dots,r_n\gamma)$ is a restricted context instantiated by terms in $sat(\phi)$. By definition of $l\rightarrow r$ where $r$ is obtained from $l$ via root and leaf permutations, the size of the context part of $r\gamma$ cannot be greater than $|C|$.

                \end{enumerate}

                Consequently, in all the three cases above, $r\gamma$ is a term $C'[S_1',\dots,S_k']$ such that $S_1',\ldots,S_k' \in sat(\phi)$, $\fn(C') \cap \tilde{n} = \emptyset$, and $|C'| \leq |C|$.
        \end{proof}

	\subsection{Local Stability}\label{sec:localstability}

        We can now show that any contracting convergent TRS has the local stability property introduced in Definition~\ref{def:locally-stable}.
For the frame saturation given in Definition~\ref{def:sat}, it suffices to check that all the items defining local stability (cf Definition~\ref{def:locally-stable}) are satisfied.
	
	\begin{thm}\label{thm:locally-stable}
		Any contracting convergent TRS is locally stable.
	\end{thm}
	\begin{proof}
		The first two conditions follow from Definition~\ref{def:sat} where $sat(\phi)$ is given in the particular case of a contracting convergent TRS. Then, the third condition follows from Lemma~\ref{lemma:small-context}. 
		The final condition follows from Lemma~\ref{lemma:sat-finite}.
	\end{proof}
	
	\noindent
	Directly from Theorem~\ref{thm:locally-stable} and Theorem~\ref{thm:know-by-local-stability}, we obtain the following corollary.
	
	\begin{cor}\label{cor:deduction-static-equiv}
		The deduction and static equivalence problems are both decidable for the class of contracting convergent TRSs.
	\end{cor}
          The deduction and static equivalence problems are thus decidable for any TRS from Example~\ref{example:cap-contracting}.

	\section{Relation to Existing Notions and Properties}\label{sec:relations}
	In this section we consider an additional knowledge problem, the cap problem~\cite{DBLP:conf/rta/AnantharamanNR07}, the relation of contracting TRS to the finite variant property~\cite{Comon-LundhD05}, and the relation to the layered convergent property~\cite{DBLP:journals/tocl/BaudetCD13}. 
	
	\subsection{The Cap Problem}\label{sec:relations-cap}
	The cap problem~\cite{DBLP:conf/rta/AnantharamanNR07,Anantharaman2007HAL} is another knowledge problem introduced to model an intruder's ability to obtain access to something that was intended to be secret. We show that the cap problem can be considered as a particular deduction problem.
	
	\begin{defi}[Cap Problem]
          Let $\Sigma$ be a signature containing a constant, $m$, called the \emph{secret}. A subset $\Sigma_{ir} \subseteq \Sigma \setminus \{m\} $ of publicly known symbols from the signature $\Sigma$ is called the \emph{intruder repertoire}. Symbols not in the intruder repertoire are called \emph{private}. The intruder repertoire is said to be \emph{complete} if $\Sigma_{ir} = \Sigma \setminus \{m\}$. A cap term is a linear term built over $\Sigma_{ir}$. Given a set $S$ of ground terms built over $\Sigma$ such that at least one of the terms contains $m$, and a convergent TRS, $R$, over $\Sigma \setminus \{m\}$, the \emph{cap problem} asks if there exist a cap term $t$ with $\Var(t) = \{ x_1,\dots,x_n\}$ and ground terms $s_1,\dots,s_n \in S$ such that $t \{ x_1 \mapsto s_1,\dots,x_n \mapsto s_n \} \downarrow_{R} = m$.
	\end{defi}
	
	As shown below, the cap problem can be easily related to the deduction problem, provided that the intruder repertoire is complete.
 
	\begin{lem}\label{lem:cap-problem}
	 If the intruder deduction problem is decidable for a convergent TRS $R$, then the cap problem with a complete intruder repertoire is also decidable for $R$.
	\end{lem}
	\begin{proof}
	  Consider the frame $\phi = \nu \{ m \}.\sigma$ where $m$ is the constant denoting the secret and $\sigma = \{ x_1 \mapsto s_1, ~\ldots, ~x_n \mapsto s_n \}$ where $S = \{ s_1,\dots,s_n\}$ is the set of ground terms given by the cap problem. Then, following the respective definitions, there exists a recipe $\zeta$ such that $\zeta \sigma \downarrow_{R} = m$ iff there exists a cap term for the cap problem. 
	\end{proof}

	Applying Theorem~\ref{thm:locally-stable} together with the above lemma, we obtain the following.
	\begin{cor}
The cap problem with a complete intruder repertoire is decidable for any contracting convergent TRS.
	\end{cor}
	
	\begin{rem}
		As illustrated above one of the differences between the cap problem and deduction is that the cap problem allows for a restricted signature over which the intruder must work. The deduction problem assumes that essentially the entire signature, minus some restricted names, is public. However, it should be possible to extend the definition of deduction to allow for a restricted intruder repertoire. This new definition would then encapsulate the current deduction definition and the cap problem. Indeed, although this would need to be shown, it seems like the current deduction procedures could work for this new definition by restricting their saturation procedures to building terms only over the intruder repertoire.  
	\end{rem}

	\subsection{Relation to the  Finite Variant Property}\label{sec:properties}
	The \emph{Finite Variant Property} (FVP) is a useful property which is utilized in a number of applications, including protocol analysis. The FVP was introduced in~\cite{Comon-LundhD05} (see also~\cite{Bouchard13, DBLP:journals/jlp/EscobarSM12} for additional examples in using the FVP).
	It has been shown~\cite{Comon-LundhD05} that a TRS has the FVP iff it has
	the following boundedness property.
	
	\begin{defi}[Boundedness Property]\label{def:bounded}
		A convergent TRS, $R$, has the \emph{boundedness property} if 
		$\forall t \exists n \forall \sigma ~: ~t(\sigma\da) \rightarrow_{R}^{\leq n} (t \sigma) \da$. That is, for any term $t$ there exists a bound, $n$, on the number of step required
		to reach the normal form, and this bound is independent of the substitution. 
	\end{defi}

	 One could naturally ask if the graph-embedded or contracting definitions just lead to systems with the FVP.  This is not the case but some of the contracting systems listed in Example~\ref{example:cap-contracting}, such as blind signatures, do have the FVP as shown below. This is not surprising, given that the FVP can be useful for proving properties like termination. Another interesting question could be: are there meaningful examples from the protocol analysis literature for which deduction and static equivalence are decidable, do not have the FVP, but are representable by contracting convergent TRSs? Here we answer this question positively.  
	
	\begin{exa}\label{example:add2}
		Consider again the theory of Addition, $R_{add}$, from Example~\ref{example:add}.
		$R_{add}$ is a contracting convergent TRS, is locally stable, and contains no $AC$-symbols, thus deduction and static equivalence are decidable.
		However, $R_{add}$ does not have the FVP, we can see this by considering the rule $\plus(x, s(y))  \rightarrow \plus(s(x), y)$ and the boundedness property. Notice that 
		for any finite bound $n$ one can
		select a normal form substitution, $\sigma$, such that 
		$\plus(x, s(y))\sigma \overset{> n}{\rightarrow_{R_{add}}} (\plus(x, s(y))\sigma)\da$. Namely, $\sigma = \{ y \mapsto s^{n+1}(z) \}$. Since 
		$R_{add}$ does not have the boundedness property it cannot have the FVP~\cite{Comon-LundhD05}. Yet, $R_{add}$ is a contracting convergent TRS. Notice that the second and third rules are already subterm. The first rule is obtained by applying Definition~\ref{def:permutative}. Therefore, $R_{add}$ satisfies Corollary~\ref{cor:deduction-static-equiv}.
	\end{exa}

As a positive result, there is an easy way to get a strictly contracting TRS with the FVP:

\begin{lem}
\label{lem:sufficient-cond-FVP}
Assume $R$ is any strictly contracting convergent TRS such that for any rule $l \rightarrow r$ in $R$, $r$ is either a variable or a non-variable term rooted by a constructor symbol. Then, $R$ has the FVP. 
\end{lem}

\begin{proof}
Consider any rule $l \rightarrow r$ in $R$. By assumption, $R$ is strictly contracting, and so $r$ is of depth at most $1$. If $r$ is variable, there is no way to overlap at a non-variable position in $r$ with the left-hand side of any rule in $R$. Assume now $r$ is of depth $1$. If there is a non-variable overlap in $r$ with the left-hand side of any rule in $R$, this overlap can only occur at the root position of $r$. Since the function symbol at the root position of $r$ is necessarily a constructor symbol, it cannot occur at the root-position of any left-hand side of $R$. Since there is no non-variable overlap in $r$ with the left-hand side of any rule in $R$, this is simple way for $R$ to satisfy the property of being forward-closed as defined in~\cite{Bouchard13}. Then, relying on the fact that a TRS has the FVP iff it has a finite forward closure~\cite{Bouchard13}, we can conclude that $R$ has the FVP.
\end{proof}

\begin{exa}
For the theory of blind signatures and the theory of prefix with pairing (Example~\ref{example:cap-contracting}), the corresponding TRSs satisfy the assumption of Lemma~\ref{lem:sufficient-cond-FVP}. Consequently, these two strictly contracting convergent TRSs have the FVP.

For the theory of trap-door commitment, the strictly contracting convergent TRS obtained  by considering additional rules (see Example~\ref{example:not-cap-contracting}) does not satisfy the assumption of Lemma~\ref{lem:sufficient-cond-FVP}.
\end{exa}

	\subsection{Relation to the Layered Convergent Property}\label{sec:relations-layered}
	The \emph{layered convergent} property is yet another useful property for TRS modeling security protocols because it can be used to ensure that the  YAPA~\cite{DBLP:journals/tocl/BaudetCD13} tool for protocol analysis does not fail.
	YAPA is a tool for computing intruder knowledge in the formal analysis of security protocols and it is able to work with a large number of protocol specifications. While YAPA allows us to consider private symbols in addition to public ones, we assume here that all the function symbols are publicly known. In this classical setting, the tool is not guaranteed to terminate or return successfully, it could return a ``don't know'' answer. However, it is shown in~\cite{DBLP:journals/tocl/BaudetCD13} that the ``layered convergent'' property  for a given TRS, with some additional conditions, can be used to show termination. This property is defined using a particular term decomposition to express any left-hand side of the TRS. 
        
	\begin{defi}[Term Decomposition]\label{def:decomp}
		Let $n,~p$, and $q$ be non-negative integers. A $(n,p,q)$-\emph{decomposition} of a term $l$ is a context $C$ with $n+p+q$ context holes and $l = C[l_1, \ldots, l_n, ~y_1, \ldots, y_p, ~z_1, \ldots, z_q]$ where:
		\begin{itemize}
			\item $l_1, \ldots, l_n$ are mutually distinct non-variable terms,
			\item $y_1, \ldots, y_p$ and $z_1, \ldots, z_n$ are mutually distinct variables, and
			\item $y_1, \ldots, y_p \in Var(l_1, \ldots, l_n)$ but $z_1, \ldots, z_n \not\in Var(l_1, \ldots, l_n)$.
		\end{itemize}
	\end{defi}

	\begin{defi}[Layered TRS]\label{def:layer-convergent}
		A TRS, $R$, is \emph{layered} if there exists an ascending chain of sub sets 
		$\emptyset = R_{0} \subseteq R_{1} \subseteq \dots R_{n+1} = R\: (n \geq 0)$ s.t.  
		for every rule $l \rightarrow r \in R_{i+1} \backslash R_{i}$ and every $(n,p,q)$-decomposition
		$l=D[l_1, \ldots, l_m, ~y_1, \ldots, y_p, ~z_1, \ldots, z_q]$ one of the following holds:
		\begin{enumerate}
			\item $\Var(r) \subseteq \Var(l_1, \ldots, l_m)$.
			\item There exist  $C_0, C_1, \ldots C_k$ and $s_1, \ldots s_k$ such that $r = C_0[s_1, \ldots, s_k]$ and for any $j=1,\dots,k$, we have $C_j[l_1, \ldots, l_m, ~y_1, \ldots, y_p, ~z_1, \ldots, z_q] \stackrel{\epsilon}{\rightarrow}_{R_i}^{0,1} s_j$.
		\end{enumerate}
	\end{defi}

	\begin{rem}
	What is interesting here is that both the contracting definition (Definition~\ref{def:cap-contracted}) and the layered definition (Definition~\ref{def:layer-convergent}) point to the ability of accessing subterms 
	as a key component to ensuring that the knowledge problems are decidable. Indeed, there is a relation between the two definitions, as shown in the next result. 
	\end{rem}
	 
	\begin{thm}\label{thm:contracting=>layered}
		Any contracting TRS is layered. 
	\end{thm}
	\begin{proof}
          Let $R$ be a contracting TRS and $R_1 = \{ l \rightarrow r \in R ~| ~ l \rightarrow r \text{ is subterm} \}$. It has already been shown in~\cite{DBLP:journals/tocl/BaudetCD13} that any subterm convergent TRS is layered. Thus, if $R_1 = R$ we are done. Otherwise, let $R_2 = R$ and any rule $l \rightarrow r \in R_2 \backslash R_1$ must be contracting but not subterm. 
          Let $l \rightarrow r \in R_2 \backslash R_1$ and consider any arbitrary $(n,p,q)$-decomposition $l=D[l_1, \ldots, l_m, ~y_1, \ldots, y_p, ~z_1, \ldots, z_q]$. If $r \in \{l_1, \ldots, l_m \} \cup
          \{y_1, \ldots, y_p \} \cup \{ z_1, \ldots, z_q\}$, then we are done. Thus, assume that $r \neq l_i$, for any $1 \leq i \leq m$ and also $r \not\in \{y_1, \ldots, y_p \} \cup \{ z_1, \ldots, z_q\}$. Let $r = C_0[s_1, \ldots, s_k]$, we show how to construct $C_0$ and how to obtain each $s_i$, $1 \leq i \leq k$, from the decomposition. There are two cases based on $depth(r)$: 
		\begin{itemize}
			\item Assume $depth(r) = 1$.
			Then, each $s_i$ is a variable. If $s_i \in  \{y_1, \ldots, y_p \} \cup \{ z_1, \ldots, z_q\}$ for all $1 \leq i \leq k$, then by definition we are done. Thus, assume there exists at least one $s_i \not\in \{y_1, \ldots, y_p \} \cup \{ z_1, \ldots, z_q\}$, which implies that $s_i \in Var(l_j)$ for some $l_j$, $1 \leq j \leq m$, in $l$, but $s_i$ is not a subterm of $l_i$ in $r$. thus,  by Definition~\ref{def:projection-closed-deriv} there is a rule $t' \rightarrow s_i 
			\in R_1$ such that $t' = C_i[l_j]$.

			\item Assume $depth(r) = 2$. 
			Here again we can assume that each $s_i$ are variables. This is due to the fact that non-variable portions of the term can be placed in the context $C_0$. If $s_i$ is contained 
			in $ \{y_1, \ldots, y_p \} \cup \{ z_1, \ldots, z_q\}$, then by definition we are done. Otherwise, $s_{i} \in Var(l_i)$, for some $1 \leq i \leq m$, in $l$, but $s_i$ is no longer a subterm of $l_i$ in $r$. Thus,  by Definition~\ref{def:projection-closed-permut-eq} there is a rule in $R_1$ of the form $C_i[l_i] \rightarrow s_{i}$.  \qedhere
			\end{itemize}	
		
	\end{proof}
	
	\begin{rem}
		Notice that the above proof outlines a natural and simple layered system for any contracting TRS $R$. Let $R_0 = \emptyset$, 
		$R_1 = \{ l \rightarrow r \in R ~| ~ l \rightarrow r \text{ is subterm} \}$, and  
		$R_2 = R \setminus R_1$. Then $R_0 \cup R_1 \cup R_2 = R$ and any rule needed 
		to deduce variables for decompositions of $R_2$ rules will be subterm and thus contained in $R_1$. Note that there are TRSs which are layered for $n=1$ but not contracting. 
	\end{rem}
	
	Theorem~\ref{thm:contracting=>layered} implies that for contracting convergent TRSs, the YAPA procedure will not fail, i.e., return a ``don't know'' answer. However, this result does not prove termination. However, we conjecture that the YAPA procedure does terminate for contracting convergent TRSs.

       \section{Contracting Convergent Systems in Unions of Theories}\label{sec:combination}

In this section, we investigate the knowledge problems in unions of theories including at least a contracting convergent TRS. We mainly focus on constructor-sharing unions of theories, where it is possible to reuse existing modularity results and combination methods~\cite{ErbaturMR17}. Let us first consider the class of strictly contracting convergent TRSs. As shown below, this class is closed by constructor-sharing union:

       \begin{thm}
        If $R_1$ and $R_2$ are two strictly contracting convergent TRSs such that the function symbols shared by $R_1$ and $R_2$ are constructors for both $R_1$ and $R_2$, then $R_1 \cup R_2$ is strictly contracting convergent. 
       \end{thm}

       \begin{proof}
         First, note that a strictly contracting TRS is terminating. By definition, if $R_1$ and $R_2$ are strictly contracting, then so is $R_1 \cup R_2$. Thus, $R_1 \cup R_2$ is terminating. By assumption $R_1$ and $R_2$ are confluent, and so also locally confluent. The local confluence is a modular property for constructing-sharing TRSs~\cite{Middeldorp-Thesis90}. Thus, $R_1 \cup R_2$ is locally confluent. Then, $R_1 \cup R_2$ is confluent  since $R_1 \cup R_2$ is terminating. Consequently, $R_1 \cup R_2$ is both confluent and terminating, equivalently, it is convergent.
       \end{proof} 
       
       The combination framework developed in~\cite{ErbaturMR17} for the knowledge problems in unions of constructor-sharing theories can be applied to contracting convergent TRSs that are not necessarily strictly contracting:

       \begin{thm}
        If $R_1$ and $R_2$ are two contracting convergent TRSs such that the function symbols shared by $R_1$ and $R_2$ are constructors for both $R_1$ and $R_2$, then both deduction and static equivalence are decidable in $R_1 \cup R_2$. 
       \end{thm}
       \begin{proof}
       	Directly from~\cite{ErbaturMR17} where it is shown that, for constructor sharing theories, the knowledge problems are decidable for $R_1 \cup R_2$ if they are decidable for $R_1$ and $R_2$. Section~\ref{sec:localstability} establishes the decidability of the knowledge problems for contracting convergent TRSs.
       \end{proof}

       We now study the case of unions of theories defined as a contracting convergent TRS $R$ plus an equational theory $E$ that cannot be oriented as a convergent TRS. Typical examples for $E$ are provided by permutative theories. Fortunately, deduction is decidable in the class of permutative theories. 

\begin{thm}\label{thm:deduc-permute-theories}
Deduction is decidable in any permutative theory.
\end{thm}  

\begin{proof}
  Consider the problem of checking whether a term $t$ is deduced from $\phi = \nu \tilde{n}.\sigma$ modulo any permutative theory  $E$. For any term $s$ such that $|s| > |t|$, we have $s\sigma \neq_E t$: otherwise, we would have $|s\sigma| \geq |s| > |t|$ and $s\sigma =_E t$ implies $|s\sigma| = |t|$. Thus, only finitely many terms $s$ (up to a renaming) have to be considered in order to check whether $t$ is deduced from $\phi$ modulo $E$.
\end{proof}

When an equational theory $E$ is defined by a finite set of permutative axioms built using function symbols that do not occur as the root positions, these symbols can be viewed as constructors for $E$. It has been shown in~\cite{ErbaturMR17} that a permutative theory can be successfully combined with a theory given by a convergent TRS, provided that the shared symbols are constructors for both theories. Thus, directly from~\cite{ErbaturMR17} we obtain:

\begin{thm}
\label{thm:combi-deduction}
  Let $R$ be any contracting convergent TRS and $E$ any permutative presentation such that any function symbol $f$ shared by $R$ and $E$ is a constructor for $R$ and, for any axiom $l=r$ in $E$, $f$ does not occur neither as a root symbol of $l$ nor as a root symbol of $r$. Then, deduction is decidable in $R \cup E$.
\end{thm}

\begin{exa}[Intruder Theory with a Permutative Axiom]\label{example:combi}
Let us consider a theory used in practice to model a group messaging protocol~\cite{DBLP:conf/ccs/Cohn-GordonCGMM18}. For this protocol, the theory modeling the intruder can be defined~\cite{Nguyen2019} as a combination $R \cup K$ where
\[K = \{
\mathit{keyexch}(x, \mathit{pk}(x'), y, \mathit{pk}(y')) = \mathit{keyexch}(x', \mathit{pk}(x), y', \mathit{pk}(y)) \} \]
and 
\[ R  = \left\{
\begin{array}{ll}
\begin{array}{lcl}
  \mathit{adec}(\mathit{aenc}(m,\mathit{pk}(sk)),sk) & \rightarrow & m \\
  \mathit{checksign}(\mathit{sign}(m,sk),m,\mathit{pk}(sk)) & \rightarrow & ok
\end{array}
&
\begin{array}{lcl}
\mathit{getmsg}(\mathit{sign}(m,sk)) & \rightarrow & m\\
  \mathit{sdec}(\mathit{senc}(m,k),k) & \rightarrow & m
\end{array}
\end{array}
\right\}.\]

$R$ is a subterm convergent TRS and $K$ is a variable-permuting theory sharing with $R$ the constructor symbol $\mathit{pk}$. Theorem~\ref{thm:combi-deduction} applies to $R \cup K$, leading to a combined decision procedure for deduction in $R \cup K$. 
\end{exa}

Deduction is decidable in permutative theories, as stated in Theorem~\ref{thm:deduc-permute-theories}. However, the static equivalence problem is undecidable in general for permutative theories or even leaf permutative theories, as show recently in~\cite{Erbatur2024}. However, we conjecture there are useful subclasses of permutative theories, including the theory $K$ introduced in Example~\ref{example:combi}, for which static equivalence is decidable and thus for which a static equivalence form of Theorem~\ref{thm:combi-deduction} could be developed. We plan to investigate this in a future work.

	\section{Conclusions and Future Work}\label{sec:conclusions} 
	
	In this paper, we have explored the idea of graph-embedded term rewrite systems and shown their applicability in protocol analysis for identifying protocols with the local stability property. This in turn allows for the identification of protocols with decidable deduction and static-equivalence problems. We have also compared and contrasted this new definition to several other concept, including the FVP, the cap problem, and Layered TRS. Finally, we have developed several combination results for these TRS. In the first version of this paper~\cite{DBLP:conf/fscd/SatterfieldEMR23} we had identified several problems and areas for further research. Some of those problems have now been answered in this paper, in particular:
	\begin{itemize}
		\item We had conjectured in~\cite{DBLP:conf/fscd/SatterfieldEMR23} that the cap problem should be decidable for contracting systems and have show this
		to be the case, see Section~\ref{sec:relations-cap}. However, one has to be careful how the signature is defined, and we currently require it to be completely public (except for some constants).  
		\item We asked in~\cite{DBLP:conf/fscd/SatterfieldEMR23} about the relation between the graph-embedded systems and the layered systems. This has now been developed in Section~\ref{sec:relations-layered}.
	\end{itemize}
	\noindent
	We have also made a number of improvements over the presentation in~\cite{DBLP:conf/fscd/SatterfieldEMR23}:
	\begin{itemize}
		\item We have improved and simplified the definition of \emph{contracting convergent TRS}.
		\item We have improved the proof of the undecidability of the knowledge problems for graph-embedded TRS. 
	\end{itemize}
	
	While we have answer some problems from~\cite{DBLP:conf/fscd/SatterfieldEMR23}, additional problems still remain to be investigated. For example, While the knowledge problems are undecidable for graph-embedded convergent systems in general and they are decidable for contracting, there is a gap between the two classes of systems. That is, how much can the contracting subclass be extended before the undecidable barrier is encountered? With respect to graph theory ideas, we are also interested in knowing if additional graph theory ideas could be useful in symbolic security protocol analysis. For example, not absolutely all theories considered in~\cite{DBLP:journals/tcs/AbadiC06, DBLP:journals/tocl/ChadhaCCK16, DBLP:journals/jar/CiobacaDK12, Dreier2017} are graph-embedded. It would be interesting to know if such systems could be considered via graph minor concepts. Just as in~\cite{DBLP:conf/fscd/SatterfieldEMR23}, we remain interested in seeing how ideas from graph theory, translated into the TRS setting, are useful for question not only in protocol analysis but also in other rewriting domains. We had also asked about termination and if the new definition could be useful in developing termination results for procedures such as~\cite{DBLP:journals/tcs/AbadiC06, DBLP:journals/tocl/ChadhaCCK16, DBLP:journals/jar/CiobacaDK12, Dreier2017}. It remains to investigate termination in these procedures. 

	Finally, in addition to answering some of the open problems of ~\cite{DBLP:conf/fscd/SatterfieldEMR23}, we have developed a number of new results not asked about in~\cite{DBLP:conf/fscd/SatterfieldEMR23}. This includes identifying a subclass of graph-embedded systems that guarantee the FVP, see Section~\ref{sec:relations}, and developing new combination results, see Section~\ref{sec:combination}. In particular, we have investigated the union of contracting convergent systems with some permutative theories. As a future work, we plan to study how undecidability proofs known for unification in permutative theories could be adapted to static equivalence.
	
	\bibliographystyle{alphaurl} 
	\bibliography{subterms-refs}

\end{document}